\documentclass[conference, 10pt]{IEEEtran}
\IEEEoverridecommandlockouts

\usepackage{cite}
\usepackage{amsmath,amssymb,amsfonts}
\usepackage{amsthm}
\usepackage{dsfont}
\usepackage{enumitem}
\usepackage{mathtools}
\usepackage{amsthm}
\usepackage{graphicx}
\usepackage{textcomp}
\usepackage{xcolor}
\usepackage{algpseudocode}
\usepackage{enumitem}
\usepackage{tikz}
\usetikzlibrary{positioning}
\usetikzlibrary{arrows.meta, positioning}
\newcommand{\remove}[1]{}

\DeclareMathOperator*{\argmin}{argmin}

\usepackage{array}
\usepackage{multirow}
\usepackage{caption}
\usepackage{comment}
\usepackage{subcaption}
\usepackage{multicol}
\usepackage{graphicx}
\usepackage{asymptote}
\usepackage[ruled,vlined]{algorithm2e}
\usepackage{array}

\newtheorem{lemma}{Lemma}
\newtheorem{remark}{Remark}
\def\BibTeX{{\rm B\kern-.05em{\sc i\kern-.025em b}\kern-.08em
    T\kern-.1667em\lower.7ex\hbox{E}\kern-.125emX}}

\newtheorem{theorem}{Theorem}
\begin{document}
\setlength{\baselineskip}{11.7pt}
\title{Lagrange Index based Scheduling for Minimizing Age of Updates from Heterogeneous Sources}
\author{Aniket~Mukherjee, Joy~Kuri and~Chandramani~Singh
}

\maketitle
\thispagestyle{plain}
\pagestyle{plain}
\begin{abstract}
Modern sensing systems generate heterogeneous updates ranging from small status packets to large data objects. 
\remove{AoI has emerged as a fundamental performance metric for sensing systems, as it captures the timeliness of information at the receiver and directly impacts the accuracy and reliablity of real-time monitoring and control applications.} We study a single-hop wireless uplink network where sensors generate updates at will, each consisting of a sensor dependent number of packets. Under a strict medium-access constraint and non-preemptive (no-switching) transmissions, decision stages become action-dependent and stochastic. We formulate the problem as a restless multi-armed bandit (RMAB) with semi-Markov decision process (SMDP) dynamics and develop a Lagrange index based heuristic for minimizing weighted average AoI cost. For the weighted AoI setting, we utilize the structural properties of the heuristic to enable efficient index computation. Numerical results demonstrate consistent performance gains over existing non-preemptive scheduling policies, providing a practical solution for heterogeneous freshness-aware systems.
\end{abstract}
\section{Introduction}
Modern sensing systems consist of spatially distributed sensors that continuously monitor physical processes and transmit status updates to a central controller or edge server \cite{Sensing_net_tubaishat}. These updates are often heterogeneous in size, ranging from small scalar measurements to large data objects such as images, video frames, or LIDAR scans. In slotted wireless systems (eg., Wi-Fi 6 and 5G), such updates may span multiple transmission slots, and medium-access constraints limit simultaneous transmissions. As sensing networks increasingly support real-time monitoring, autonomous systems, and cyber-physical applications, efficient scheduling of heterogeneous updates becomes critical to maintaining timely situational awareness\cite{kaul_min_aoi}.

The Age of Information (AoI) quantifies information freshness by measuring the time elapsed since the most recently generated update was successfully delivered \cite{kaul_update, kosta_aoi, kaul_min_aoi}. Unlike traditional delay or throughput metrics, AoI directly captures the timeliness of the received information, making it particularly suitable for sensing and monitoring applications where stale data can degrade estimation accuracy, control performance or tracking reliability. Although AoI is a theoretical metric, its practical value has been demonstrated experimentally; for example, \cite{tripathi2023wiswarm} shows that an AoI-aware WiFi middleware significantly improves information freshness and tracking accuracy in UAV networks compared to standard WiFi-UDP/TCP. 

We consider a set of heterogeneous sources connected to a common receiver through an unreliable wireless channel that allows only one source to transmit at a time~(see Figure~\ref{fig:multi_source_channel}). The sources generate updates of different sizes, which must be delivered to the receiver so as to minimize the long term weighted average age of updates at the receiver. 
A non-preemptive~(non-switching) transmission discipline is employed where a scheduled source keeps the channel until it successfully transmits all the packet pertaining to its update.
We formulate the problem as a Semi-Markov Decision Process~(SMDP). However, this problem suffers from the curse of dimensionality. To address it, we treat the problem as a Restless Multi-Armed Bandit~(RMAB) and develop a Lagrange index based heuristic. To the best of our knowledge, this is the first work to study RMAB formulation of weakly coupled SMDPs and to develop scalable index policies for them. Our heuristic scheduling policy substantially outperforms the existing solutions. 
\vspace{-10mm}
\begin{figure}[htbp]
\centering
\includegraphics[width=0.98\linewidth]{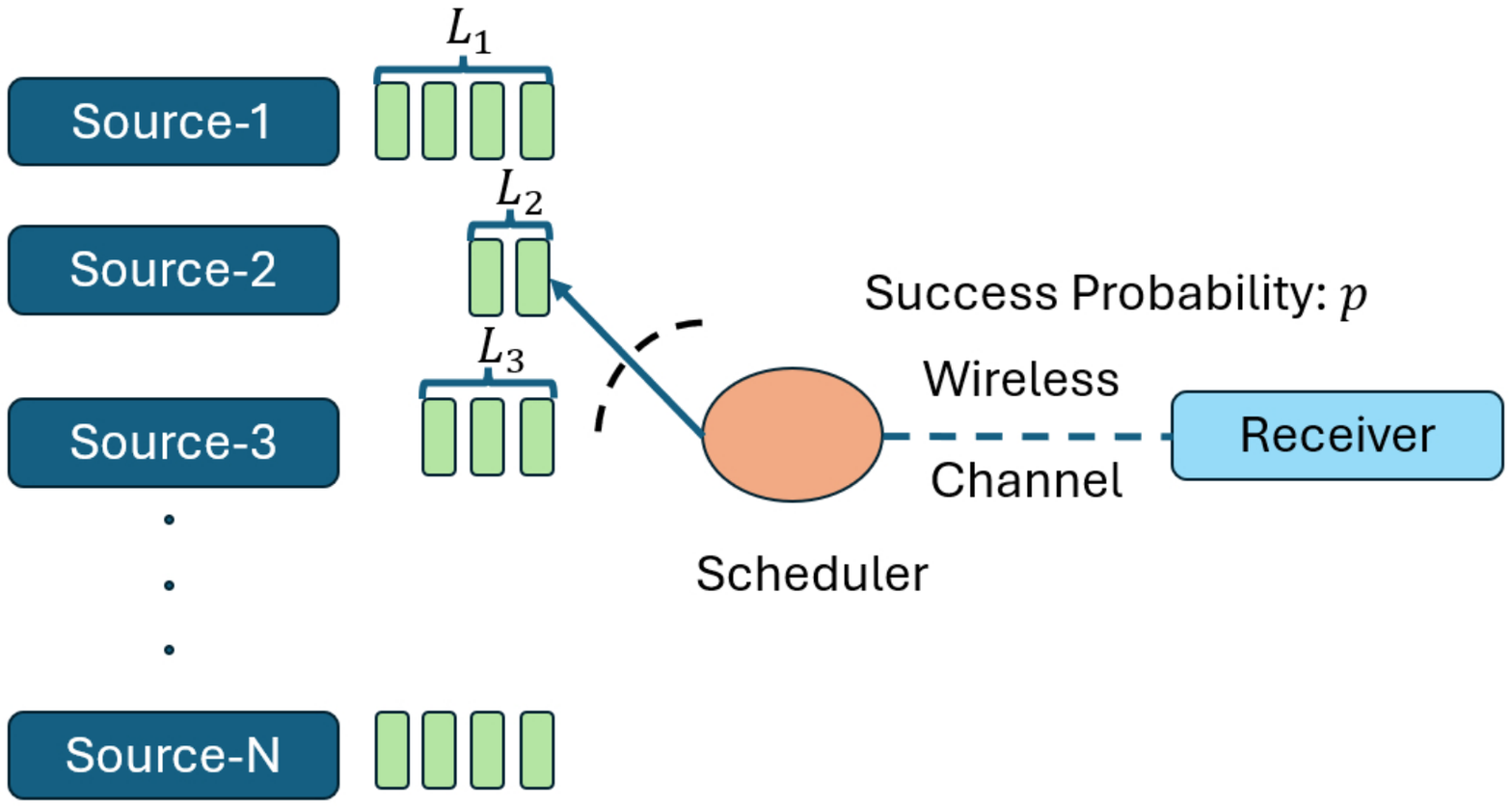}
\caption{$N$ sources with heterogeneous update sizes, connected to a receiver through an unreliable wireless channel.}
\label{fig:multi_source_channel}
\end{figure}
\vspace{-5mm}
\subsection{Related Work}
We now briefly discuss the related work on (a)~AoI minimization and on (b)~Lagrange index based heuristics for RMAB problems.
\subsubsection{AoI Minimization}
The concept of AoI was introduced in \cite{kaul_update}, and a comprehensive overview of its extensions and practical relevance was provided in \cite{yates_survey}. 
The papers concerned with AoI minimization can be categorized as in Table~{\ref{tab:related_works}}. There are two dimensions: (a) Whether all source updates are of equal length or not, and (b) whether one can interrupt transmission of an update from a source before it is complete, and ``switch'' to transmitting another source's update [``switching''], or not [``non-switching''].
\cite{kadota_aoi} studies a wireless network in which a base station generates single-packet updates and, in each slot, selects a user to which it will transmit, over an unreliable channel with user-dependent success probabilities $p_i$ (constant over time). The objective is to minimize the long-term average AoI. The greedy policy was shown to be optimal when $p_i = p$ for all $i$. For the asymmetric case, a Whittle index and a max-weight policy were proposed.  In \cite{tripathi_whittle}, for equal update lengths, the same model as in \cite{kadota_aoi} was considered, but with a nonlinear function of AoI as the objective. The authors established indexability and proposed a Whittle index–based policy. In reality update lengths can be different and we consider differing update lengths in our work.

For heterogeneous sources with reliable channels, \cite{tripathi_pross_compute} considered an AoI model in which each source requires a fixed processing time, followed by a fixed transmission time under a non-preemptive service discipline. They established indexability and proposed a Whittle index policy.  
The problem was formulated as a restless multi-arm bandit and decoupled into single-source Markov Decision Process (MDP). In the resulting single-source MDP, only two actions were considered: when the tagged source is scheduled, the AoI resets to the total processing and transmission time; when it is not scheduled, the AoI increases by one. In heterogeneous systems, this assumption is restrictive, since the cost when the tagged source is not scheduled depends on service duration of the scheduled source.

In \cite{zhou_hetero_aoi} and \cite{zhao}, the same uplink model  as ours (Figure~\ref{fig:multi_source_channel}) was considered. \cite{zhou_hetero_aoi} considers policies that allowed switching and suggested a suboptimal policy to minimize the average AoI. For the non-preemptive setting, \cite{zhao} proposed a stationary randomized policy, referred to as the No-Switching Randomized Policy (NSRP). 

\remove{In heterogeneous non-preemptive systems such as ours, the AoI increment of a tagged source depends explicitly on the service time of the scheduled source, which varies across sources. This coupling across sources fundamentally alters the structure of the single-source problem.}
\begin{table}[htbp]
\caption{Related Works}
\centering
\begin{tabular}{|c|c|c|}
\hline
\textbf{} & \textbf{Equal length update} & \textbf{Unequal length update} \\
\hline
{\textbf{Switching}}&

&{\cite{zhao}, \cite{zhou_hetero_aoi}}\\
\hline
\textbf{Non-switching} & \cite{kadota_aoi}, \cite{tripathi_whittle} & \cite{zhao}, \cite{tripathi_pross_compute} \\
\hline
\end{tabular}
\label{tab:related_works}
\end{table}
\subsubsection{Lagrange Indices for RMABs}
Lagrange index policies for restless multi-armed bandits have been studied through linear programming relaxations in \cite{Verloop_lp, Gast_2024}, where the dual decomposition yields index-type policies and asymptotic optimality was established for large-scale systems with underlying Markov Decision Process (MDP) dynamics. The framework was further extended in \cite{Borkar_avrachenkov} to settings with unknown transition kernels, preserving asymptotic optimality via learning-based methods. Our setting differs fundamentally in that we consider restless bandits under an SMDP formulation induced by non-preemptive service with random completion times, rather than discrete-time MDP dynamics. While we adopt a similar Lagrange relaxation principle, we do not establish asymptotic optimality for the resulting policy under SMDP and instead use it as a computationally tractable heuristic for the heterogeneous AoI problem.

None of the above works considers RMAB formulation and index-based policies for weakly coupled SMDPs. 
In the context of AoI minimization for heterogeneous sources and unreliable channels, existing work is limited to state-independent stationary policies that are understandably suboptimal. 
 

\subsection{Our Contributions}
Following are our main contributions.
\begin{enumerate}
    \item We model the problem of scheduling of updates from heterogeneous sources to minimize  the long run weighted average age of updates at the receiver~(Section~\ref{systemmodel}). We pose it as a  average cost SMDP problem (Section~\ref{subsec:Sysmod_smdp_form}).  
    \item We propose an approach to frame general weakly coupled SMDPs as RMABs. Using this formulation, we develop a Lagrange index based heuristic for weakly coupled SMDPs~(Section~\ref{Sec:RMABSMDP}). \remove{This development extends the applicability of RMABs and index policies beyond weakly-coupled MDPs and  is of interest in its own right.}
    \item We apply the above heuristic to the update scheduling problem. In particular, we  argue that the policies suggested for the decoupled single source problems are {\it one step lookahead policies}  \cite[Section~4.4]{Bertsekas} and provide an iterative algorithm to obtain these. We use  solutions to the decoupled problems to obtain Lagrange indices, and subsequently, a heuristic for the original problem (Section~\ref{Sec:LagrangeIndexpolicy}). Our numerical results show that the proposed heuristic noticeably outperforms the known solutions (Section~\ref{sec:numerical_results}).
\end{enumerate}

\section{System Model} \label{systemmodel}
We consider a discrete-time sensing and communication system with $N$ sources~(e.g., sensors)  indexed by $1,\cdots,N$. Each source generates updates that have to be communicated to a common receiver through an unreliable wireless channel. We consider scheduling of update transmissions so as to minimize the average age of updates at the receiver. We now formally describe the update generation and communication processes, age evolution at the receiver, and the optimal scheduling problem. 

\subsubsection*{Update generation}
Source $i$'s updates are $L_i$ packet long. A source can take multiple slots to transmit an update. When a source is scheduled to transmit, it generates a new update unless it has an unfinished update transmission. 

\subsubsection*{Update transmission} The communication channel allows only one source to transmit at a time. When scheduled, a source transmits one packet per slot. Owing to channel unreliability, each packet transmission succeeds with probability $p \in (0,1]$ and fails with probability $1-p$.
Clearly, source $i$ requires at least $L_i$ slots to transmit an update. We consider non-preemptive update transmissions; a scheduled source keeps the channel until it successfully transmits all the packet pertaining to its update.   

Let $\tau^{(i)}_k$ denote the time when $k$th successful update delivery of source $i$ is accomplished. Moreover, let $X^{(i)}_k$ denote the number of slots needed for delivery of this update. From the above discussion, the generation time of this packet is $\tau^{(i)}_k - X^{(i)}_k$. Note that because of the packet generation model (generate at will) assumed, if $L_i=1$ then $X^{(i)}_k = 1$ for all $k$ and if $L_i \ge 2$ then  $X^{(i)}_k, k \geq 1$ are i.i.d. random variables. In particular,
\[
\mathbb{P}(X^{(i)}_k = l) = p^{(i)}_l
\]
where, if $L_i = 1$ then $p^{(i)}_1 = 1$, and if $L_i \ge 2$ then
\begin{equation}
\label{eq:first_success_duration_distribution}
p^{(i)}_l =
\begin{cases}
\binom{l-2}{L_i-2} p^{L_i-1} (1-p)^{l-L_i}
& \text{ if } l \ge L_i,\\
0 & \text{ otherwise.}
\end{cases}
\end{equation}
The above equation follows from the observation that after successful transmission of the first packet, the remaining $L_i -1$ successful packet transmissions require a
negative binomial number of additional slots.
It follows that $\mathbb{E}[X^{(i)}_k] = \frac{L_i - (1-p)}{p}$.

\subsubsection*{Age of updates} For any source, the age of updates, also referred to as the age of information~(AoI), at the receiver is defined as the time elapsed since generation of the last received update. Let $\bar{v}_i(t)$ denote the age of updates of source $i$ at the receiver at time $t$. Clearly,
\[
\bar{v}_i(t) = \begin{cases}
    X^{(i)}_k & \text{ if } t = \tau^{(i)}_k, \\
    \bar{v}_i(t-1) + 1 & \text{ otherwise.}
\end{cases}
\]
The expected long term weighted average cost is given as 
\[
\lim_{T \to \infty}\frac{1}{T}\mathbb{E}\left[\sum_{t=1}^T \alpha_i \bar{v}_i(t)\right].
\]
\paragraph*{Optimal Scheduling problem}
Update transmissions of various sources must be scheduled so as to minimize the above cost. Notice that a source must be selected for update transmission after either of the following events.
\begin{enumerate}
    \item An update is successfully delivered. 
    \item Transmission of the first packet of an update fails.
\end{enumerate}
Clearly, the intervals between successive decision instants are random variables. We frame the above scheduling problem as a semi-Markov decision process~(SMDP) as described below.

\subsection{SMDP Formulation}
\label{subsec:Sysmod_smdp_form}
Let $t_k$, $k\geq 1$ denote the successive decision instants. Let $v_i(k)$ denote the age of updates of source $i$ at the receiver at $t_k$; $v(k) \coloneqq (v_1(k),  \cdots, v_N(k)) \in \mathbb{Z}_+^N$ represents the state at $t_k$. Further, let  $a(k) \coloneqq (a_1(k), \cdots, a_N(k))\in \{0,1\}^N$ denote the scheduling decision at $t_k$; $a_i(k) =1$ if source $i$ is chosen for update transmission at $t_k$ and $a_i(k) =0$ otherwise. The scheduling constraint prescribes that exactly one source  be chosen at each $t_k$, i.e., $\sum_{i=1}^N a_i(k) = 1 \qquad \forall k \ge 1$.
Let $\Delta(k) \coloneqq t_{k+1} - t_k, k \geq 1$ denote the gaps between successive decision instants.
Given action $a(k)$ with $a_i(k)=1$, $\Delta(k) = 1$ if the first packet's transmission at $t_k$ fails and $\Delta(k) \ge L_i$ otherwise. More specifically, $\Delta(k)$ is distributed as follows.  
\begin{equation}
\mathbb{P}(\Delta(k) = l|a_i(k)=1) = pp^{(i)}_l + (1-p)\mathds{1}(l=1).     
\end{equation}
It can be easily checked that $\mathbb{E}[\Delta(k)|a_i(k)=1] = L_i$. More generally, the expected value of the interval $\Delta(k)$, denoted as $S(a(k))$, is given as follows.
\[S(a(k)) \triangleq \mathbb{E}[\Delta(k)] = \sum_{j=1}^N a_j(k)L_j. \]

Further, given $v(k)$ and $a(k)$ with $a_i(k)=1$, 
$v_i(k+1) = v_i(k)+1$ if the first packet's transmission at $t_k$ fails and $v_i(k+1) = \Delta(k)$ otherwise. In either case,  $v_j(k+1) = v_j(k)+\Delta(k)$ for all $j \neq i$. See Figure~\ref{fig:smdp_transition_twoflows} for an illustration. We thus see that $v_j(k+1) = v_j(k) + 1$ for all $j$ with probability $1-p$, and $v_i(k+1) = l$ and $v_j(k+1) = v_j(k)+l$ for all $j \neq i$ with probability $pp^{(i)}_l$. Clearly, $\bar{v}_i(t), t \geq 1$ is a controlled semi-Markov chain~(or SMDP). We refer to $t_k$ as the $k$th stage of this SMDP.   

Now, we describe the single stage cost associated with the above SMDP. The weighted update-age cost incurred over $\{t_k,\cdots,t_{k+1}-1\}$, also referred to as the $k$th stage cost, is a function of the state $v(k)$, action $a(k)$ and the interval length $\Delta(k)$. 
Note that the ages of updates for all the sources increment by one at successive slots between $t_k$ and $t_{k+1}$. Hence the $k$th stage cost associated with source $i$ equals 
\[
\alpha_i\sum_{s=0}^{\tau(k)-1} (v_i(k)+s) 
 = \alpha_i\left(\Delta(k) v_i(k) + \frac{\Delta(k)(\Delta(k)-1)}{2}\right).\]
Consequently, the expected $k$th stage cost associated with source $i$, denoted as $g_i(v_i(k),a(k))$, is given by
\begin{equation}
g_i(v_i(k),a(k))
= 
\alpha_i \sum_{j=1}^N a_j(k)(
v_i(k)  L_j + w(L_j)).
\label{eq:expected_cost_weighted_aoi}
\end{equation}
where
\begin{equation}
w(L_j)\coloneqq \mathbb{E}\left[\frac{\Delta(k)(\Delta(k)-1)}{2} l\bigg \vert a_j(k)=1\right] = \frac{L_j(L_j-1)}{2p}.
\label{eq:expected_length_weighted_aoi}
\end{equation}

\begin{figure}[htbp]
\centering
\setlength{\tabcolsep}{2pt}

\begin{tabular}{cc}
\begin{tikzpicture}[scale=0.62, >=Latex,
    axis/.style={thick, ->},
    flowi/.style={thick, red},
    flowj/.style={thick, blue},
    every node/.style={font=\scriptsize}
]
\draw[axis] (0,0) -- (3,0);
\draw[axis] (0,0) -- (0,5);
\node[above] at (0,5) {AoI};
\node[right] at (3,0) {Time};

\foreach \y in {1,...,4}
    \draw[dashed, very thin, gray!30] (0,\y) -- (3,\y);
\foreach \x in {1,...,2}
    \draw[dashed, very thin, gray!30] (\x,0) -- (\x,4);

\node[below] at (0,0) {$t_k$};
\node[below] at (1,0) {$t_{k+1}$};

\draw[flowi] (0,3) -- (1,3);
\draw[flowi] (1,3) -- (1,4);
\node[left, red] at (0,3) {$v_j(k)$};
\node[right, red] at (1,4) {$v_i(k)+1$};

\draw[flowj] (0,1) -- (1,1);
\draw[flowj] (1,1) -- (1,2);
\node[left, blue] at (0,1) {$v_i(k)$};
\node[right, blue] at (1,2) {$v_i(k)+1$};
\end{tikzpicture}
&
\begin{tikzpicture}[scale=0.62, >=Latex,
    axis/.style={thick, ->},
    flowi/.style={thick, red},
    flowj/.style={thick, blue},
    every node/.style={font=\scriptsize}
]
\draw[axis] (0,0) -- (4,0);
\draw[axis] (0,0) -- (0,6);
\node[above] at (0,6) {AoI};
\node[right] at (4,0) {Time};

\foreach \y in {1,...,5}
    \draw[dashed, very thin, gray!30] (0,\y) -- (4,\y);
\foreach \x in {1,...,3}
    \draw[dashed, very thin, gray!30] (\x,0) -- (\x,5);

\node[below] at (0,0) {$t_k$};
\node[below] at (3,0) {$t_{k+1}=t_k+l$};

\draw[flowi] (0,3) -- (1,3);
\draw[flowi] (1,3) -- (1,4);
\draw[flowi] (1,4) -- (2,4);
\draw[flowi] (2,4) -- (2,5);
\draw[flowi] (2,5) -- (3,5);
\node[left, red] at (0,3) {$v_i(k)$};
\node[right, red] at (3,5) {$v_j(k)+3$};

\draw[flowj] (0,1) -- (1,1);
\draw[flowj] (1,1) -- (1,2);
\draw[flowj] (1,2) -- (2,2);
\draw[flowj] (2,2) -- (2,3);
\draw[flowj] (2,3) -- (3,3);
\draw[flowj] (3,3) -- (3,1.5);
\node[left, blue] at (0,1) {$v_i(k)$};
\node[right, blue] at (3,1.5) {$3$};

\draw[dashed] (0,0.4) -- (3,0.4);
\node at (2.0,0.65) {\scriptsize transmission};
\end{tikzpicture}
\\
{\scriptsize (a) $\Delta(k)=1$ (failure)}
&
{\scriptsize (b) $\Delta(k) = 3$ (success)}
\end{tabular}

\caption{State transition from $t_k$ to $t_{k+1}$ given that source $i$ is scheduled at $t_k$. We assume $L_i = 2$.}
\label{fig:smdp_transition_twoflows}
\end{figure}
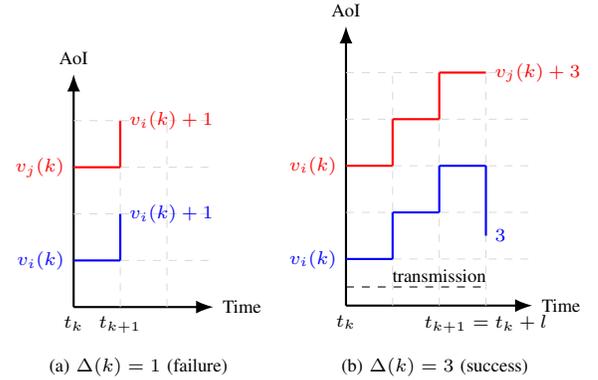

We can formulate the scheduling problem as a long-term average cost control problem. A stationary admissible policy is a mapping from the state space $\mathbb{Z}_+^N$ to the set of $N$-dimensional standard unit vectors. So, given the initial state, any policy $\pi$ induces a set of actions $a(k) \in \{0,1\}^N, k \geq 1$ such that 
\begin{equation}
\sum_{j=1}^N a_j(k)=1 \quad \forall k \geq 1.
\label{eq:single_arm_constraint}
\end{equation}
Using Markov renewal reward theorem~\cite[Theorem~7.5]{ross_applied}, the long-term expected average-cost associated with policy $\pi$ is given by
\begin{equation}
\lim_{K\to\infty}
\frac{
\mathbb{E}_\pi\left[
\sum_{k=0}^K
\sum_{i=1}^N
g_i(v_i(k),a(k))
\right]
}{
\mathbb{E}_\pi\left[
\sum_{k=0}^K
\sum_{i=1}^N a_i(k)L_i
\right]
}
\label{eq:primal_problem}
\end{equation}
The optimal control problem aims to minimize this cost over all the admissible policies. Evidently, this problem suffers from  curse of dimensionality and is non-viable even for moderate number of sources. 

We address this issue via posing the scheduling problem as a RMAB problem with the sources being treated as the arms, and proposing a Lagrange index based heuristic.  
We can see that the update age evolution of different sources are weakly coupled through the actions at the decision instants as in RMAB problems. 
However, unlike a classical RMAB setting where states constitute a controlled Markov chain, state evolution in our case is a controlled semi-Markov chain. Here, the decoupled problems associated with different arms are also SMDPs.    
To the best of our
knowledge, RMAB formulation of SMDPs has not been
considered so far. In the next section, we discuss how the RMAB framework can be extended for general weakly-coupled SMDPs and also propose a Lagrange index based heuristic.




\section{RMAB Formulation of SMDPs}
\label{Sec:RMABSMDP}

In this section, we consider general weakly-coupled SMDPs, treat them as RMABs and develop a Lagrange index based heuristic.  We retain most of the notation in Section~\ref{subsec:Sysmod_smdp_form}. In particular, we consider $N$ arms with $v_i(k) \in \mathcal{V}_i$ being the state of arm $i$ and $a_i(k) \in \{0,1\}$ being the action associated with arm $i$ at the $k$th stage. As before, the actions are coupled as $\sum_{i=1}^N a_i(k) = 1$ for all $k \geq 1$ and. $v(k)$ and $a(k)$ represent state and action vectors at the $k$th stage. Given $a(k),  \geq 1$, the states of different arms, i.e., $v_i(k), k\geq 1$ evolve independently. In particular, for all $k\geq 1$, given $v_i(k)$ and $a(k)$, $v_i(k+1)$ is independent of $v_j(k),  j \neq i$ and previous state and action vectors. For any $i$, $L_i$ represents the sojourn time of the joint SMDP in the $k$th stage given $a_i(k) = 1$. Finally, $g_i(v_i,a)$ represents the single stage cost associated with arm $i$ given that it is in state $i$ and action $a$ is taken.  The total single state cost for state-action pair $(v,a)$ is given by $\sum_{i=1}^N g_i(v_i,a)$. We do not specify the state transition probabilities as these are not needed for  the following discussion.

An admissible policy $\pi:\prod_{i=1}^N \mathcal{V}_i \to \{0,1\}^N$ induces a sequence of actions $a(k), k\geq 1$ that satisfy~\eqref{eq:single_arm_constraint}. Let $\Pi$ be the set of all admissible policies. We aim to minimize the long-term average expected cost, given by~\eqref{eq:primal_problem}, over all $\pi \in \Pi$. Constraint~\eqref{eq:single_arm_constraint} which couples the individual SMDPs necessitates that these be solved together rendering an optimal solution intractable. Below, we propose a novel relaxation of~\eqref{eq:single_arm_constraint} that facilitates decoupling of individual SMDPs.

\subsection{A Relaxed Problem}
We adopt the standard relaxation used in classical RMABs~\cite{Verloop_lp,Gast_2024,Borkar_avrachenkov}. Let $N(T)$ denote the number of decision stages up to time $T$. Constraint~\eqref{eq:single_arm_constraint} which ensures that exactly one arm be played at each stage also implies that the long term expected fractions of slots during which different arms  are played add up to one. More formally,~\eqref{eq:single_arm_constraint} implies that    
\begin{equation}
\lim_{T\to\infty}
\frac{1}{T}
\mathbb{E}_\pi
\!\left[
\sum_{k=0}^{N(T)} \sum_{i=1}^N a_i(k)S(a(k))
\right]
=1.
\label{eq:relaxed_constraint_time}
\end{equation}
Note that the expected sojourn times $\mathbb{E}_\pi\left[\sum_{i=1}^N a_i(k)L_i \right]$ are bounded by $\max_i L_i$.  Hence,
by Markov renewal reward theorem~\cite{Bertsekas,ross_applied},~\eqref{eq:relaxed_constraint_time} is equivalent to
\begin{equation}
\lim_{K\to\infty}
\frac{
\mathbb{E}_\pi
\left[
\sum_{k=0}^K \sum_{i=1}^N  a_i(k)S(a(k))
\right]
}{
\mathbb{E}_\pi
\left[
\sum_{k=0}^K S(a(k))
\right]
}
=1.
\label{eq:relaxed_problem_constraint}
\end{equation}
The relaxed problem therefore becomes
\begin{equation}
\min_{\pi}
\lim_{K\to\infty}
\frac{
\mathbb{E}_\pi
\left[
\sum_{k=0}^K
\sum_{i=1}^N
g_i(v_i(k),a(k))
\right]
}{
\mathbb{E}_\pi
\left[
\sum_{k=0}^K S(a(k))
\right]
}
\label{eq:relaxed_problem}
\end{equation}
over  all $\pi$ for which the the induced actions $a(k), k \geq 0$ satisfy~\eqref{eq:relaxed_problem_constraint}. 
 Constraint~\eqref{eq:relaxed_problem_constraint} allows multiple arms to be played in a slot. 
This relaxation enables a Lagrangian formulation that decomposes into $N$ independent single-arm control problems.

\subsubsection*{The Dual Problem}
Introducing a Lagrange multiplier $\lambda\in\mathbb{R}$ associated with constraint~\eqref{eq:relaxed_problem_constraint}, 
the Langrangian can be written as 
\[
\lim_{K\to\infty}
\frac{
\mathbb{E}_\pi
\left[
\sum_{k=0}^K
\sum_{i=1}^N
(g_i(v_i(k),a(k))
+\lambda a_i(k)S(a(k)))\right]
}{
\mathbb{E}_\pi
\left[
\sum_{k=0}^K S(a(k))
\right]
}
-\lambda
\]
So, the dual problem is
\begin{equation}
\max_{\lambda\in\mathbb{R}}
\left(\min_{\pi}
\sum_{i=1}^N \tilde{\mathcal{L}}_i(\pi,\lambda) -\lambda\right).
\label{eq:dual_problem}
\end{equation}
where
\begin{align}
\lefteqn{
\tilde{\mathcal{L}}_i(\pi,\lambda) \coloneqq  \nonumber } \\
&\lim_{K\to\infty}
\frac{
\mathbb{E}_\pi\Bigg[
\sum_{k=0}^K
\Big(
g_i(v_i(k),a(k))
+
\lambda\, a_i(k)\,S(a(k))
\Big)
\Bigg]}{
\mathbb{E}_\pi\left[\sum_{k=0}^K S(a(k))\right]}.
\label{eq:lagrangian_reduced}
\end{align}

\remove{
We now describe an approach to solve the relaxed problem. Subsequently, this solution is used to define Lagrange indices.
}
\remove{
we form the
Lagrangian by adding a penalty proportional to the constraint violation. Since the resulting term $-\lambda \tilde{S}(a(k))$ appears in the numerator with the same
normalization $\mathbb{E}_\pi[\sum_{k=0}^K \tilde{S}(a(k))]$, it contributes the constant $-\lambda$ to the average-cost objective. Therefore, for a fixed $\lambda$, minimizing the Lagrangian over $\pi$ is equivalent to minimizing the following modified average-cost ratio:
\begin{equation}
\begin{aligned}
\tilde{\mathcal{L}}(\pi,\lambda)
\triangleq\;
\lim_{K\to\infty}
\left(
\mathbb{E}_\pi\!\left[\sum_{k=0}^K \tilde{S}(a(k))\right]
\right)^{-1}
\\
\;
\sum_{i=1}^N
\mathbb{E}_\pi\!\Bigg[
\sum_{k=0}^K
\Big(
g_i(v_i(k),a(k))
+
\lambda\, a_i(k)\,\tilde{S}(a(k))
\Big)
\Bigg].
\end{aligned}
\label{eq:lagrangian_reduced}
\end{equation}

The corresponding dual problem is
\begin{equation}
\max_{\lambda\in\mathbb{R}}
\;\;
\min_{\pi\in\Pi}
\tilde{\mathcal{L}}(\pi,\lambda).
\label{eq:dual_problem}
\end{equation}
}

We note that \eqref{eq:dual_problem} does not decouple. However, a lower bound to \eqref{eq:dual_problem} can be obtained as follows
\begin{equation}
\max_{\lambda\in\mathbb{R}}\left(\sum_{i=1}^N
\min_{\pi}
\tilde{\mathcal{L}}_i(\pi,\lambda) -\lambda\right).
\label{eq:dual_problem_lower_bound}
\end{equation}
Now we propose a method to solve \eqref{eq:dual_problem_lower_bound}.
\remove{
We solve the relaxed problem~(\eqref{eq:relaxed_problem} subject to~\eqref{eq:relaxed_problem_constraint}) using a dual update method where (a)~we minimize the Lagrangian for any given $\lambda$ (b)~we use this solution to iterate dual variable $\lambda$ such that the iterations converge to the optimal value $\lambda^\ast$. Solution to  Lagrangian minimization for $\lambda^\ast$ yields the optimal policy to the relaxed problem.    
}
\subsubsection{Inner Minimization Problem} For a fixed $\lambda$ The inner minimization problem  is equivalent to 
\[
\min_{\pi} \tilde{\mathcal{L}}_i(\pi,\lambda)
\]
which decouples into $N$ SMDPs associated with the $N$ arms. A policy $\pi$  for the $i$th arm's problem is a mapping from $\mathcal{V}_i$ from $\{0,1\}^N$.
\remove{For any fixed $\lambda$, the inner objective admits a separable structure across arms, which allows the minimization over $\pi$ to be carried out via $N$ single-flow subproblems.} The collection of these single-flow solutions provides the minimizing policy for the current $\lambda$.
 We now focus on the corresponding single-arm problem. In particular, for a fixed multiplier $\lambda$, we study the optimal control of an  individual arm $i$ under the modified expected per-stage cost
$g_i(v_i(k),a(k))+\lambda a_i(k)\tilde{S}(a(k))$. 

\remove{
In the single-arm problem, the per-stage cost depends explicitly on the chosen
action. This is because different actions induce different holding-time distributions between decision stages, which in turn change both the expected
stage duration and the cumulative cost incurred during that stage. Therefore,
the resulting control problem is a semi-Markov decision process (SMDP) with
action-dependent costs and transitions.
}
The single-arm average-cost optimality
equation can be written as
\begin{equation}
h_i(v_i)
=
\min\Big\{
Q_{i,i}(v_i),\;
\min_{j\neq i} Q_{i,j}(v_i)
\Big\},
\label{eq:single_flow_bellman_general}
\end{equation}
where
\begin{align*}
\lefteqn{Q_{i,k}(v_i)} \\
& =
\begin{cases}
g_i(v_i,\mathbf{e}_i)
+(\lambda-\theta(\lambda))S(\mathbf{e}_i)
+\mathbb{E}[h_i(V_i')\mid v_i,\mathbf{e}_i],
k=i,
\\
g_i(v_i,\mathbf{e}_k)
-\theta(\lambda) S(\mathbf{e}_k)
+\mathbb{E}[h_i(V_i')\mid v_i,\mathbf{e}_k],
k\neq i.
\end{cases}
\label{eq:Qi_compact}
\end{align*}
Here, $\mathbf e_k$ denotes the unit scheduling vector that selects source $k$. In \eqref{eq:single_flow_bellman_general}, $h_i(v)$ is the bias (relative
value) function and $\theta(\lambda)$ is the optimal average cost for the single-flow problem. The Bellman equation contains $N$ actions corresponding to activating each arm. Solving \eqref{eq:single_flow_bellman_general} for a fixed $\lambda$
yields the optimal policy for source $i$ under the Lagrangian relaxation. The assumptions on $g_i(v, \mathbf{e}_i)$ such that the solution to \eqref{eq:single_flow_bellman_general} exists is as per \cite{Puterman_book, sennott1998stochastic}.

\subsubsection{Dual Update} 
Note that the dual problem is a concave maximization problem.
The resulting dual function is then maximized over $\lambda$ using standard
scalar ascent methods (e.g., subgradient ascent), yielding an optimal multiplier
$\lambda^\star$, which is then used to update the dual variable. Iterating this procedure produces $\lambda^\star$ and a corresponding policy. To update the dual variable, we compute the average fraction of time the $i$-th
arm is activated under this optimal policy. This is obtained by solving an
auxiliary average-cost Bellman equation where the per-stage cost equals $1$
when arm $i$ is activated and $0$ otherwise. Let $\mu_i(\lambda)$ denote the resulting long-run activation fraction of arm $i$ under the optimal policy for multiplier $\lambda$. The dual variable is then updated using bisection method.


\subsection{Lagrange Indices}
Upon convergence to $\lambda^\star$, we compute, for each arm $i$, the Lagrange index
\begin{equation}
\gamma_i(v_i)
=
Q_{i,i}(v_i)
-
\min_{j\neq i} Q_{i,j}(v_i).
\end{equation}

\subsubsection*{Heuristic for the Original Problem}

The resulting Lagrange index policy selects the  arm
\begin{equation}
m^*
=
\argmin_{i} \gamma_i(v_i).
\label{eq:lag_index_arm}
\end{equation}

This constitutes the Lagrange index policy for restless multi-armed bandits for SMDP. In Section~\ref{sec:singleflowprob} we discuss about the AoI minimization problem stated in 
\eqref{eq:primal_problem}.

\section{Single source AoI minimization} 
\label{sec:singleflowprob}

We now return to the weighted AoI minimization problem introduced in
Section~\ref{systemmodel}. Under Lagrangian relaxation, the original
multi-source problem decouples into independent single-source subproblems as per Section~\ref{Sec:RMABSMDP}.
We focus on one such subproblem corresponding to a fixed source index $i$.
\subsection{Average-Cost Bellman Equation (SMDP)}
Fix a multiplier $\lambda \in \mathbb{R}$. For the weighted AoI cost, the expected cumulative cost incurred over a decision stage when scheduling source $j$ admits the closed-form expressions given in
\eqref{eq:expected_length_weighted_aoi}--\eqref{eq:expected_cost_weighted_aoi}.
The corresponding average-cost optimality equation is given by
\cite{Puterman_book}:
\begin{equation}
\begin{aligned}
h_i(v_i)
=
(1-p)\,h_i(v_i+1)
+
\min\Big\{
Q_{i,i}(v_i),\;
\min_{j\neq i} Q_{i,j}(v_i)
\Big\}.
\end{aligned}
\label{eq:smdp_single_flow_bellman_Q}
\end{equation}

The term $(1-p)h_i(v_i+1)$ represents the continuation value under transmission failure, which occurs with probability $1-p$ regardless of the chosen action. In that case, the AoI deterministically increases by
one. Since this transition is action-independent, the corresponding term
appears outside the minimization.
The action-dependent terms are
\begin{equation*}
Q_{i,j}(v_i) \triangleq
\begin{cases}
g_i(v_i,\mathbf{e}_i)
+ (\lambda-\theta(\lambda))L_i
+ \displaystyle\sum_{l\ge L_i} pp_l^{(i)} h_i(l),
 j=i, \\[0.8em]
g_i(v_i,\mathbf{e}_j)
- \theta(\lambda) L_j
+ \displaystyle\sum_{l\ge L_j} pp_l^{(j)} h_i(v_i+l),
 j\neq i.
\end{cases}
\label{eq:Q_definitions_cases}
\end{equation*}
\remove{Upon successful transmission, the duration of the decision epoch is a random variable taking values $l \ge L_j$ with probability mass function $\{pp_l^{(j)}\}$
If $j=i$, the AoI resets to $l$ after success.If $j\neq i$, the AoI evolves to $v_i + l$.} Here, $\mathbf e_k$ denotes the unit scheduling vector that selects
source $k$. The scalar $\theta(\lambda)$ denotes the optimal average cost of the
single-source problem for a fixed multiplier $\lambda$. The multiplier $\lambda$ can be interpreted as a penalty you pay to associated with scheduling source $i$.

We establish a few structural properties of the $h_i(.)$ function. In Lemma~\ref{lem:h_monotone_full}, we show the monotonicity property of the Bellman equation. In Lemma~\ref{lem:dominant_competitor} we reduce the action space from $N$ to $2$ actions. Theorem~\ref{thm:threshold_structure} shows that the 2 action Bellman equation has a threshold structure using \emph{one step lookahead policies}.
\begin{lemma}[Monotonicity of the value function]
\label{lem:h_monotone_full}
The one-step cost $g_i(v,\mathbf a)$ is nondecreasing in $v$
for every feasible action $\mathbf a$.
Consequently, any solution $h_i(\cdot)$ of
\eqref{eq:smdp_single_flow_bellman_Q}
is nondecreasing in $v$, i.e.,
\[
v_1 \ge v_2
\quad \Longrightarrow \quad
h_i(v_1) \ge h_i(v_2).
\]
\end{lemma}
\begin{proof}
    See Appendix~\ref{app:bias_monotone}
\end{proof}
We now show that the inner minimization over $j\neq i$ in \eqref{eq:smdp_single_flow_bellman_Q} reduces to a single
dominant competing action.

\begin{lemma}[Dominant competing action]
\label{lem:dominant_competitor}
Fix a source $i$ and multiplier $\lambda$, and define
\[
m(i) = \arg\min_{j \neq i} L_j.
\]
Then, for all $v_i \in \mathbb{Z}_+$ and all $j \neq i$,
\[
Q_{i,m(i)}(v_i)
\le
Q_{i,j}(v_i).
\]
That is, among all competing sources, it is optimal to select the one with the smallest update length.
\end{lemma}
\begin{proof}
The proof is provided in Appendix~\ref{app:dominant_competitor}.
\end{proof}
Lemma~\ref{lem:dominant_competitor} formalizes the structural property that,
whenever source $i$ is not scheduled, it is optimal to schedule the
``fastest'' competing source, i.e., one with smallest update length $L_j$.
Intuitively, this minimizes the time during which the age of flow $i$ continues to grow.

As a consequence of Lemma~\ref{lem:dominant_competitor}, 
the inner minimization over $j \neq i$ in 
\eqref{eq:smdp_single_flow_bellman_Q} 
is achieved by the single index $m$. 
Accordingly, the Bellman equation simplifies to a reduced 
two-action form.
\begin{equation}
\begin{aligned}
h_i(v_i)
=
(1-p)\,h_i(v_i+1)
+
\min\Big\{
Q_{i,i}(v_i),\;
Q_{i,m}(v_i)
\Big\},
\end{aligned}
\label{eq:two_action_bellman}
\end{equation}
where $Q_{i,i}$ and $Q_{i,m}$ are given in
\eqref{eq:smdp_single_flow_bellman_Q}.
We show that the Bellman equation \eqref{eq:two_action_bellman} admits a threshold structure: there exists a threshold $T_i(\lambda)$ such that it is optimal to schedule source $i$ when $v_i \ge T_i(\lambda)$ and the competing source $m$ otherwise.
\subsection{Computing the threshold and average cost}
In this section we suggest a numerical method to compute the optimal threshold of the single source problem.
\begin{theorem}[Threshold structure and characterization]
\label{thm:threshold_structure}
Fix a source index $i$ and multiplier $\lambda \in \mathbb{R}$. Consider the corresponding two-action SMDP Bellman equation  \eqref{eq:two_action_bellman}, where the competing action is denoted by 
$m \neq i$. Then the optimal policy is of threshold type: there exists a threshold 
$T_i(\lambda) \in \mathbb{Z}_{\ge 0}$ such that
\begin{equation}
\pi_i^\star(v_i)
=
\begin{cases}
\mathbf{e}_m, & L_i \le v_i < T_i(\lambda),\\[4pt]
\mathbf{e}_i, & v_i \ge T_i(\lambda).
\end{cases}
\label{eq:threshold_policy}
\end{equation}
$T_i(\lambda)$ is given by
\begin{equation}
T_i(\lambda)
=
\left\lceil
\frac{\theta(\lambda)}{\alpha_i}
-
\frac{L_m - 1}{2p}
-
\frac{L_i}{p}
\right\rceil.
\label{eq:threshold_formula_theta}
\end{equation}
\end{theorem}
\begin{proof}
 See Appendix~\ref{app:threshold_formula_theta}
\end{proof}
Theorem~\ref{thm:threshold_structure} shows that the threshold
$T_i(\lambda)$ is determined by the (unknown) average cost $\theta(\lambda)$. We now
derive a second relation between $\theta(\lambda)$ and $T_i(\lambda)$ that enables
efficient computation of both quantities for a fixed multiplier $\lambda$.

From Theorem~\ref{thm:threshold_structure} and \eqref{eq:two_action_bellman}, we can write for all $v_i \ge L_i$ as
\begin{equation}
h_i(v_i)=f_{i,1}(v_i)-\theta(\lambda) f_{i,2}(v_i),
\label{eq:h_decomposition}
\end{equation}
\remove{
for all $v_i\ge T_i(\lambda)$ the optimal action is $\mathbf{e}_i$, and the bias function admits the affine form
\begin{equation}
h_i(v_i)
=
\frac{\alpha_i L_i}{p}\,v_i
-
\frac{\theta(\lambda) L_i}{p},
\qquad v_i\ge T_i(\lambda).
\label{eq:h_affine_region}
\end{equation}
From equations  \eqref{eq:two_action_bellman} and \eqref{eq:h_affine_region},it follows that for all $v_i \ge L_i$, the bias function can be expressed as 
\begin{equation}
h_i(v_i)=f_{i,1}(v_i)-\theta_{\lambda} f_{i,2}(v_i),
\qquad v\ge L_i,
\label{eq:h_decomposition}
\end{equation}
}
where $f_{i,1}$ and $f_{i,2}$ are functions independent of $\theta(\lambda)$. In
particular, for the region $v_i\ge T_i(\lambda)$, we have
\begin{equation}
f_{i,1}(v_i)=\frac{\alpha_i L_i}{p}\,v_i,
\qquad
f_{i,2}(v_i)=\frac{L_i}{p}.
\label{eq:f_boundary_region}
\end{equation}
Next, consider the region $L_i\le v_i < T_i(\lambda)$, where the threshold policy schedules the competing source $m$. Define $ c(v_i)\coloneqq T_i(\lambda)-v_i$, so $c(v_i) \ge 1$ in this region. The Bellman equation
\eqref{eq:two_action_bellman} reduces to
\begin{equation}
\begin{aligned}
h_i(v_i)
=&
(1-p)\,h_i(v_i+1)
+
(\alpha_i v_i-\theta(\lambda))L_m
+
\alpha_i w(L_m)
+\\&
\sum_{l\ge L_m} pp_l^{(m)}\,h_i(v_i+l),
\label{eq:bellman_region_k}
\end{aligned}
\end{equation}
Since $T_i(\lambda)$ is the threshold, the term $h_i(v_i+l)$
must be treated differently depending on whether
$v_i+l < T_i(\lambda)$ or $v_i+l \ge T_i(\lambda)$.
For $l \ge c(v)$ we have $v_i+l \ge T_i(\lambda)$, so $h_i(v_i+l)$ is affine. Substituting into the Bellman equation and collecting the coefficients of $\theta(\lambda)$ gives the recursions for $f_{i,1}(\cdot)$ and $f_{i,2}(\cdot)$ in the region $L_i \le v_i < T_i(\lambda)$.
\begin{equation}
\begin{aligned}
f_{i,1}(v_i)
&=
\alpha_i v_i L_m
+
\alpha_i w(L_m)
+
(1-p)\,f_{i,1}(v_i+1)
+\\&
\sum_{L_m\le l\le c(v)} p_l^{(m)}\,f_{i,1}(v_i+l)
+
\sum_{l>c(v)} p_l^{(m)}\,
\frac{\alpha_i L_i}{p}(v_i+l),
\end{aligned}
\label{eq:f1_recursion}
\end{equation}
and
\begin{equation}
\begin{aligned}
f_{i,2}(v_i)
&=
L_m
+
(1-p)\,f_{i,2}(v_i+1)
+\\&
\sum_{L_k\le l\le c(v)} p_l^{(m)}\,f_{i,2}(v_i+l)
+
\sum_{l>c(v)} p_l^{(k)}\,
\frac{L_i}{p}.
\end{aligned}
\label{eq:f2_recursion}
\end{equation}
From equation \eqref{eq:f_boundary_region} and the
recursions \eqref{eq:f1_recursion}--\eqref{eq:f2_recursion} we uniquely determine
$f_{i,1}(v_i)$ and $f_{i,2}(v_i)$ for all $v_i\ge L_i$.
Finally, we obtain a closed-form expression for $\theta(\lambda)$ using \eqref{eq:h_decomposition};
\begin{equation}
\begin{aligned}
\lefteqn{\theta(\lambda)
=}\\
&
\frac{
p\!\left(
\lambda L_i
+ p \sum_{l\ge L_i} p_l^{(i)} f_{i,1}(l)
+ \alpha_i w(L_i)
\right)
+ \alpha_i L_i(1-p)
}{
p^2 \sum_{l\ge L_i} p_l^{(i)} f_{i,2}(l)
}.
\label{eq:theta_closed_form}
\end{aligned}
\end{equation}
The equations \eqref{eq:threshold_formula_theta} and \eqref{eq:theta_closed_form}
define an implicit relationship between the threshold $T_i(\lambda)$ and the
unknown average cost $\theta(\lambda)$. We compute $(T_i(\lambda),\theta_{\lambda})$ via a fixed-point iteration as shown in Algorithm~\ref{alg:threshold_fixed_point}.

\begin{lemma}[Monotonicity in $\lambda$]
\label{lem:monotonicity_in_lambda}
For each source $i$, the threshold $T_i(\lambda)$ and the corresponding average cost $\theta(\lambda)$ are nondecreasing functions of the multiplier $\lambda$. That is, for any $\lambda_1 \le \lambda_2$,
\[
T_i(\lambda_1) \le T_i(\lambda_2)
\quad \text{and} \quad
\theta(\lambda_1) \le \theta(\lambda_2).
\]
\end{lemma}
\begin{proof}
The proof is provided in Appendix~\ref{app:monotonicity_in_lambda}.
\end{proof}

\begin{algorithm}[]
\footnotesize
\caption{Fixed-point iteration to compute $T_i(\lambda)$}
\label{alg:threshold_fixed_point}

\KwIn{$i,\lambda,m\neq i$, $(p,\alpha_i,L_i,L_m)$, $\beta\in(0,1)$, 
$\theta(\lambda)^{(0)}$, $\varepsilon>0$}
\KwOut{$T_i(\lambda)$ and $\theta(\lambda)$}

$n \leftarrow 0$\;

\Repeat{$|\theta{\lambda}^{(n+1)}-\theta(\lambda)^{(n)}|\le \varepsilon$}{

  $T_i^{(n)} \leftarrow
  \left\lceil
  \frac{\theta(\lambda)^{(n)}}{\alpha_i}
  -\frac{L_m-1}{2p}
  -\frac{L_i}{p}
  \right\rceil$\;

  Compute $f_{i,1}(\cdot), f_{i,2}(\cdot)$ using 
  \eqref{eq:f1_recursion}--\eqref{eq:f2_recursion}\;

  $\bar{\theta}(\lambda)^{(n)} \leftarrow 
  \theta(\lambda)\!\left(T_i^{(n)}\right)$ using 
  \eqref{eq:theta_closed_form}\;

  $\theta(\lambda)^{(n+1)} \leftarrow 
  \beta\,\theta(\lambda)^{(n)} 
  + (1-\beta)\,\bar{\theta}(\lambda)^{(n)}$\;

  $n \leftarrow n+1$\;
}

$T_i(\lambda)\leftarrow T_i^{(n-1)}$\;
$\theta(\lambda)\leftarrow \theta(\lambda)^{(n)}$\;

\end{algorithm}

\section{Lagrange Index policy}
\label{Sec:LagrangeIndexpolicy}
 \subsection{Dual update via average activation fractions}

For a fixed multiplier $\lambda$, the dual problem \eqref{eq:dual_problem}
decomposes into $N$ independent single-source SMDPs, each solved as in Section~\ref{sec:singleflowprob}. To solve the outer maximization over $\lambda$ in \eqref{eq:dual_problem}, we update $\lambda$ using bisection method. The derivative of the Lagrangian with respect to $\lambda$ is given by $(\sum_{i=1}^N \mu_i(\lambda) - 1)$,
where $\mu_i(\lambda)$ denotes the long-run fraction of time (in slots)
during which source $i$ is scheduled under the single-flow optimal policy for
multiplier $\lambda$.

To compute $\mu_i(\lambda)$, we consider the policy $\pi_i^*$ in \eqref{eq:threshold_policy} and evaluate the long run fraction of time for which the action $i$ is chosen. This is done as follows; The one-step cost function is modified to
\begin{equation}
\hat{g}_i(v_i, \pi_i^*(v)) = \begin{cases}
    1, \quad v_i \ge T_i(\lambda)\\
    0, \quad L_i \le v_i < T_i(\lambda)
\end{cases}    
\end{equation}
and we consider the policy evaluation equation 
\begin{equation}
    A_i^{\pi_i^*}(v_i) = \left(\hat{g}_i(v_i, \pi_i^*(\lambda)) - \mu_i(v)\right)L_i + \mathbb{E}[A_i^{\pi_i^*}(v_i')]
    \label{eq:bias_function_mu}
\end{equation}
where $A_i^{\pi_i^*}(\cdot)$ is the bias function and $v_i'$ is the next state to which the process transitions.



From \eqref{eq:bias_function_mu}, we can write for all $v\ge L_i$,
\begin{equation}
A_i^{\pi_i^*}(v_i)=A_{i,1}(v_i)-\mu_i(\lambda)\,A_{i,2}(v_i),
\label{eq:A_decomposition}
\end{equation}
where $A_{i,1}$ and $A_{i,2}$ are independent of $\mu_i(\lambda)$.
Substituting \eqref{eq:A_decomposition} into \eqref{eq:bias_function_mu} we get
linear recursions for $A_{i,1}$ and $A_{i,2}$ over the region
$L_i\le v_i <T_i(\lambda)$, with boundary conditions for $v_i\ge T_i(\lambda)$ given
by,
$A_{i,1}(v_i)=
A_{i,2}(v_i)=\frac{L_i}{p}.$ Finally, we get the activation fraction in closed form
\begin{equation}
\mu_i(\lambda)
=
\frac{
\sum_{l\ge L_i} p_l^{(i)}A_{i,1}(l)
}{
\sum_{l\ge L_i} p_l^{(i)}A_{i,2}(l)
}.
\label{eq:theta_ai_ratio}
\end{equation}






Having computed $\mu_i(\lambda)$ policies for a fixed multiplier
$\lambda$, we update $\lambda$ to solve the outer maximization in
\eqref{eq:dual_problem} using bisection method as given in Algorithm~\ref{alg:bisection_lagrange}. 
\remove{
where $\{\eta_n\}$ is a stepsize sequence chosen to ensure convergence. Since the dual optimization is one-dimensional, an alternative is to use a bisection search over $\lambda$.} 
\begin{remark}[On existence of an exact multiplier]
In general, there may not exist a multiplier $\lambda^\star$ such that
\[
\sum_{i=1}^N \mu_i(\lambda^\star)=1.
\]
This is due to the discrete nature of the actions, which makes the mapping
$\lambda \mapsto \sum_{i=1}^N \mu_i(\lambda)$ piecewise constant and
possibly discontinuous. Consequently, the relaxed constraint may not be met
with equality for any single value of $\lambda$.
Nevertheless, there exist multipliers $\lambda_-^\star$ and $\lambda_+^\star$
such that
\[
\sum_{i=1}^N \mu_i(\lambda_-^\star) < 1
\qquad\text{and}\qquad
\sum_{i=1}^N \mu_i(\lambda_+^\star) > 1,
\]
which ensures that the dual optimality condition is satisfied. The resulting
multiplier is therefore dual optimal. However, since the feasible set of the
original problem does not satisfy a strict interior condition, a nonzero dual
gap may exist between the primal and dual problems.
\end{remark}
The
activation fractions satisfy a monotonicity property: as $\lambda$ increases,
each $\mu_i(\lambda)$ decreases, which makes bisection algorithm particularly effective.

\subsection{Lagrange index policy}
Having obtained the optimal multiplier $\lambda^\star$, we define the
Lagrange index for each source $i$ at state $v_i$ as
\begin{equation}
\gamma_i(v_i)
\triangleq
Q_{i,i}(v_i)-Q_{i,m}(v_i),
\end{equation}
which measures the incremental cost incurred by serving source $i$ over competing action. The resulting scheduling policy selects, at each decision
instant, the source with the smallest Lagrange index given in \eqref{eq:lag_index_arm}.
For classical restless bandit problems formulated as discrete-time MDPs, Lagrange index policies are known to be asymptotically optimal as
$N$ increases~\cite{Gast_2024}. 
Such guarantees have not been established for semi-Markov decision
processes (SMDPs). We adopt the Lagrange index policy as a heuristic and evaluate it numerically, where it consistently outperforms the NSRP policy in~\cite{zhao}.

\begin{algorithm}[htbp]
\footnotesize
\caption{Compute Lagrange index}
\label{alg:bisection_lagrange}

\KwIn{$\{L_i\}_{i=1}^N$, $p$, $\lambda_{\text{low}}$, $\lambda_{\text{high}}$, $\varepsilon$}
\KwOut{$\lambda^\star$, $\{\gamma_i(\cdot)\}_{i=1}^N$}

\Repeat{$|\lambda_{\text{high}}-\lambda_{\text{low}}| < \varepsilon$}{

    $\lambda \leftarrow \dfrac{\lambda_{\text{low}}+\lambda_{\text{high}}}{2}$\;

    Compute $T_1(\lambda),\dots,T_N(\lambda)$ using Algorithm~\ref{alg:threshold_fixed_point}\;

    Obtain $\mu_1(\lambda),\dots,\mu_N(\lambda)$ using \eqref{eq:theta_ai_ratio}\;

    \If{$\sum_{i=1}^N \mu_i(\lambda) - 1 < 0$}{
        $\lambda_{\text{high}} \leftarrow \lambda$\;
    }
    \Else{
        $\lambda_{\text{low}} \leftarrow \lambda$\;
    }
}

$\lambda^\star \leftarrow \dfrac{\lambda_{\text{low}}+\lambda_{\text{high}}}{2}$\;

\For{$i=1$ \KwTo $N$}{
    $\gamma_i(\cdot) \leftarrow Q_{i,i}(\cdot) - Q_{i,j}(\cdot)$\;
}

\Return $\lambda^\star$, $\{\gamma_i(\cdot)\}_{i=1}^N$\;

\end{algorithm}

\section{Numerical Results}
\label{sec:numerical_results}
We compare the proposed Lagrange index policy with NSRP~\cite{zhao}, Greedy, Scaled Greedy, and the Whittle index policy. The Greedy policy selects the flow with the largest AoI, while Scaled Greedy selects $\arg\max_i \alpha_i v_i$. Under NSRP, a probability mass function over sources is obtained by solving the optimization problem in~\cite{zhao}, and a source is selected randomly according to this distribution. For reliable channels ($p=1$), we additionally compare with the Whittle index policy for heterogeneous sources derived in~\cite{tripathi_pross_compute}.
Figure~\ref{fig:aoi_nsrp_paper} compares the NSRP, Greedy, Scaled Greedy, and proposed Lagrange index policies as the channel reliability varies. As $p$ increases, all policies improve; however, the Lagrange index policy consistently achieves the lowest long-run average weighted AoI across the range of reliabilities.

Figures~\ref{fig:aoi_varyinglengths_comparison}, \ref{fig:aoi_varying_n_comparison}, and \ref{fig:aoi_alpha_different_comparison} examine the impact of update-length heterogeneity, system scaling, and weight asymmetry, respectively. In each case, the left subplots (Figures~\ref{fig:aoi_varyinglengths_p05}, \ref{fig:aoi_varying_n_p07}, and \ref{fig:aoi_alpha_different_p07}) correspond to unreliable channels ($p<1$) and compare NSRP, Lagrange index, Greedy, and Scaled Greedy policies, where the proposed policy consistently performs best. The right subplots (Figures~\ref{fig:aoi_varyinglengths_p1}, \ref{fig:aoi_varying_n_p1}, and Figure~\ref{fig:aoi_alpha_different_p1} corresponds to the reliable channel case ($p=1$) and additionally includes the Whittle index policy. In this regime, the proposed Lagrange index policy consistently outperforms or closely matches the Whittle index policy across all considered scenarios, demonstrating its effectiveness even in settings where Whittle indices are available. Importantly, while the Whittle index policy is currently characterized only for the reliable channel case, the Lagrange index framework naturally extends to unreliable channels ($p<1$), where no Whittle index characterization is yet available.
\vspace{-5mm}
\begin{figure}[htbp]
    \centering
    \includegraphics[width=0.9\linewidth]{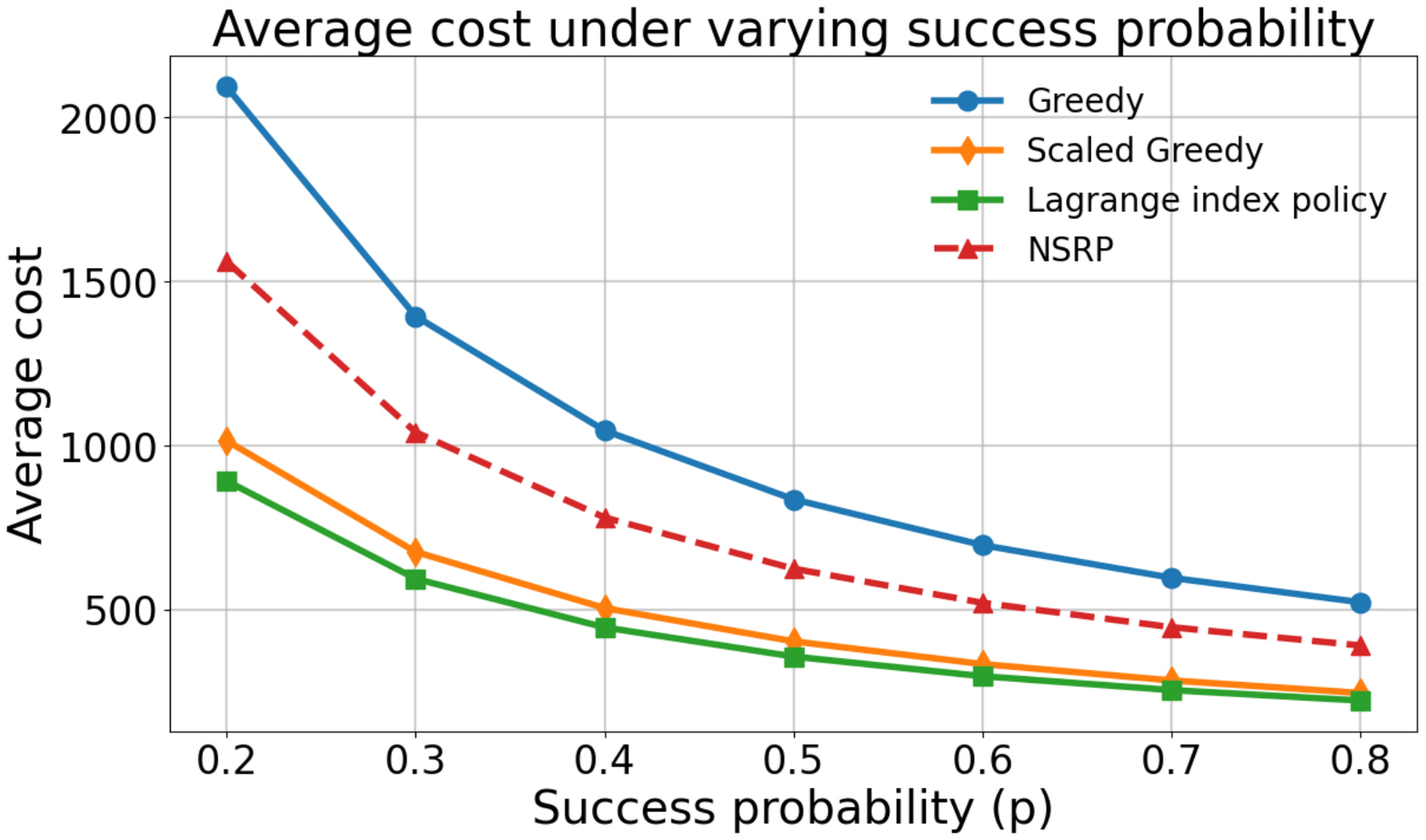}
\caption{Simulation results under varying channel reliability. The network has $N=10$ sources split evenly into two classes: Class~1 with $(L_i,\alpha_i)=(2,5)$ and Class~2 with $(L_i,\alpha_i)=(50,1)$. The common channel reliability varies over $p\in\{0.2,0.3,\ldots,0.8\}$.}
    \label{fig:aoi_nsrp_paper}
\end{figure}
\begin{figure}[]
    \centering

    \begin{subfigure}[t]{0.49\linewidth}
        \centering
        \includegraphics[width=\linewidth]{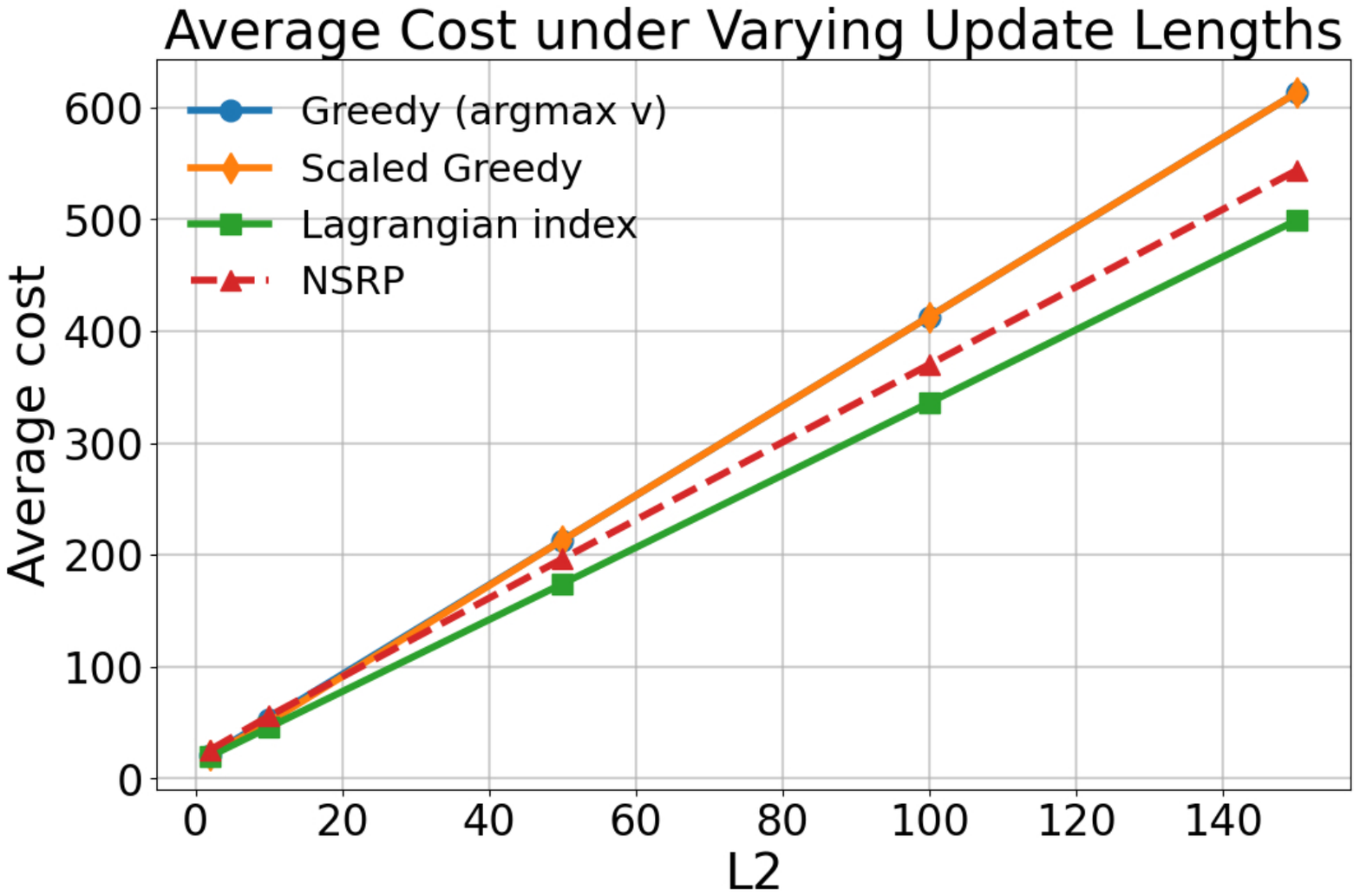}
        \caption{$p=0.5$}
        \label{fig:aoi_varyinglengths_p05}
    \end{subfigure}
    \hfill
    \begin{subfigure}[t]{0.49\linewidth}
        \centering
        \includegraphics[width=\linewidth]{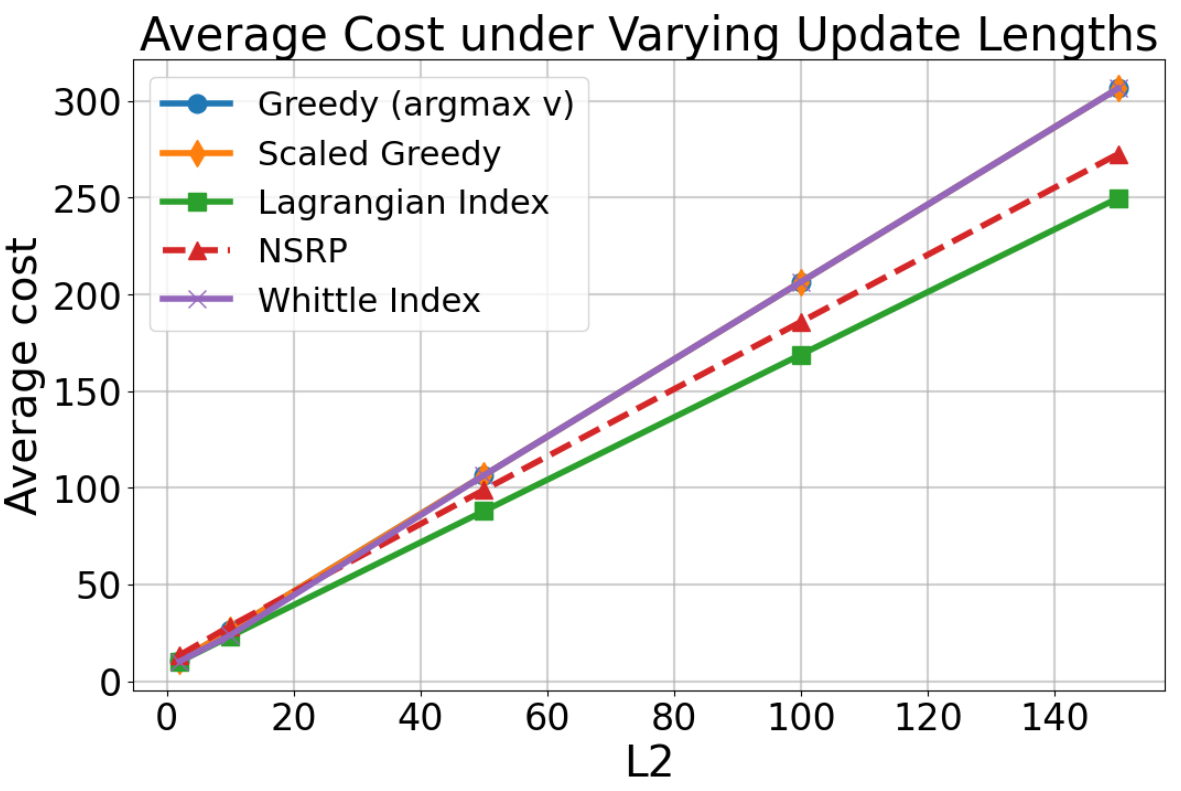}
        \caption{$p=1$}
        \label{fig:aoi_varyinglengths_p1}
    \end{subfigure}
    \caption{
    Average weighted AoI versus update-length heterogeneity. We consider $N=2$ sources with $(L_1,\alpha_1)=(2,5)$ and $(L_2,\alpha_2)=(L_2,1)$, where
$L_2 \in \{2,10,50,100,150\}$. The left panel corresponds to $p=0.5$, where the greedy and scaled greedy policies achieve identical performance. The right panel corresponds to $p=1$, where the greedy, scaled greedy, and Whittle index policies acheive identical performance.
    }
    \label{fig:aoi_varyinglengths_comparison}
    
\end{figure}
\vspace{-5mm}
\begin{figure}[]
\centering

\begin{subfigure}[t]{0.48\linewidth}
    \centering
    \includegraphics[width=\linewidth]{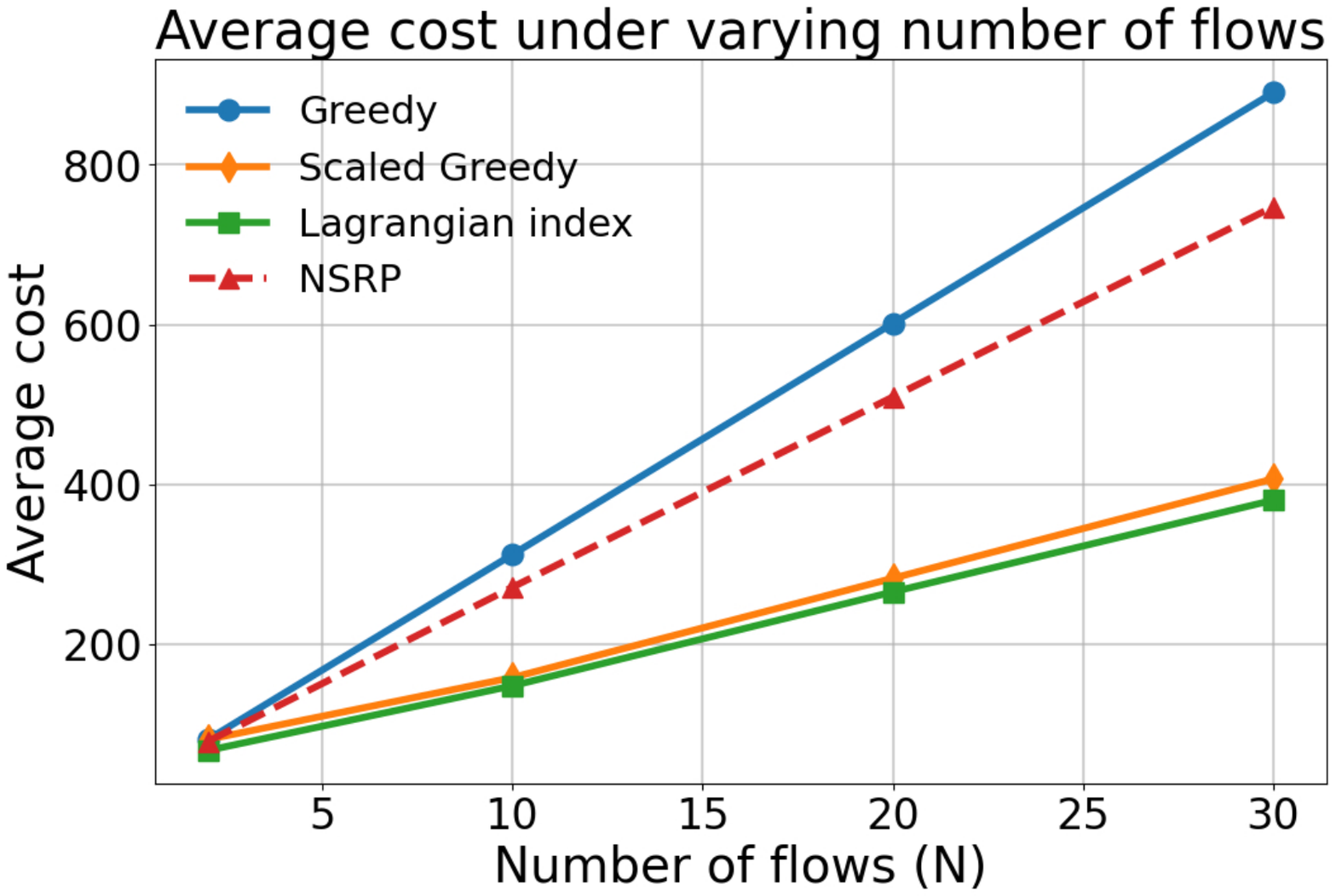}
    \caption{$p=0.7$}
    \label{fig:aoi_varying_n_p07}
\end{subfigure}
\hfill
\begin{subfigure}[t]{0.49\linewidth}
    \centering
    \includegraphics[width=\linewidth]{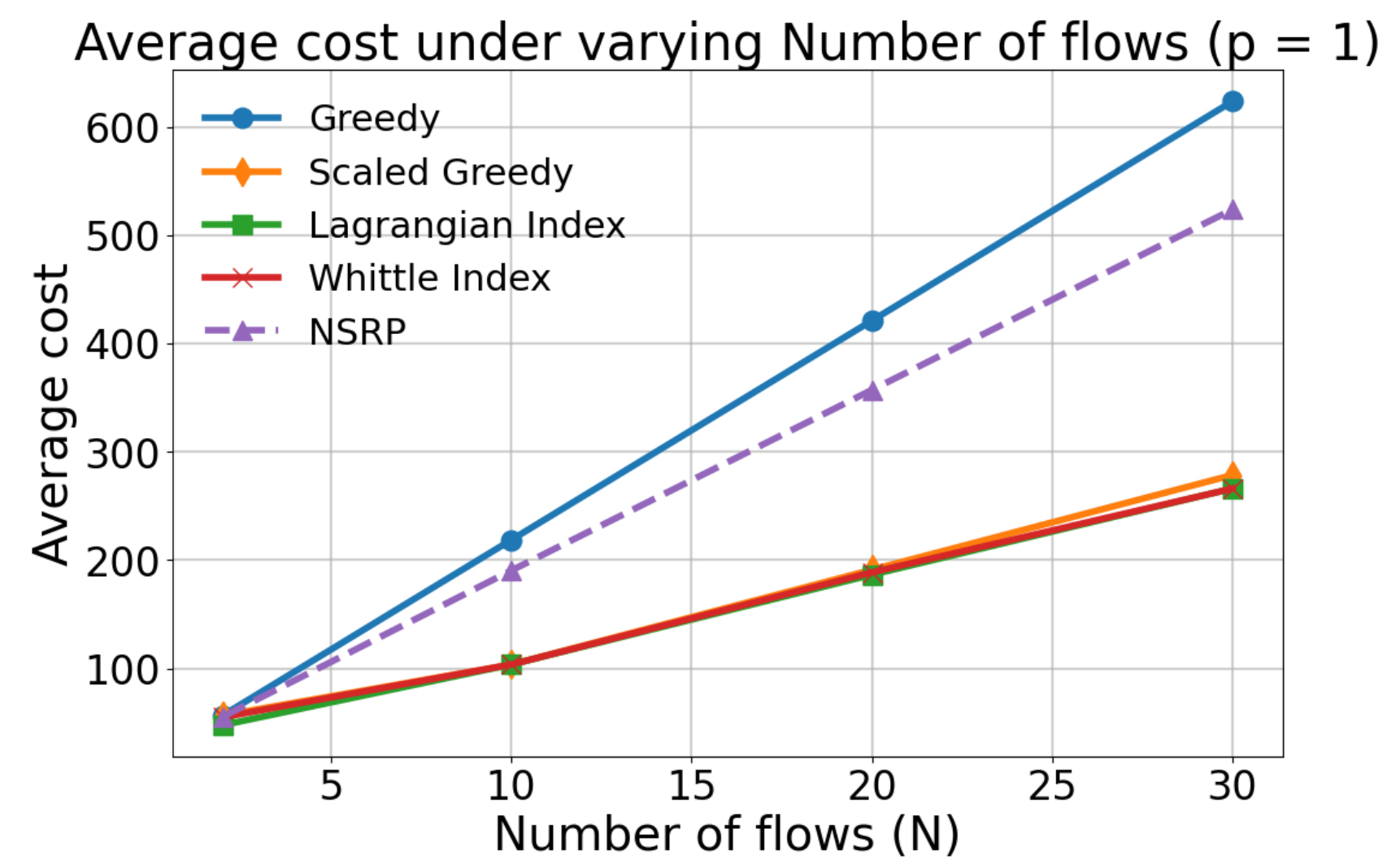}
    \caption{$p=1$}
    \label{fig:aoi_varying_n_p1}
\end{subfigure}

\caption{
Average cost versus the number of sources. We consider two classes with an equal number of sources in each class. The total number of sources is $N \in \{2,10,20,30\}$.
Class--1 sources have $(L_1,\alpha_1)=(2,5)$, and class--2 sources have $(L_2,\alpha_2)=(25,1)$. The left panel corresponds to $p=0.7$, and the right panel to $p=1$. The Whittle index and Lagrange index policies achieve identical performance.
}
\label{fig:aoi_varying_n_comparison}

\end{figure}

\begin{figure}[t]
    \centering
    \begin{subfigure}[t]{0.49\linewidth}
        \centering
        \includegraphics[width=\linewidth]{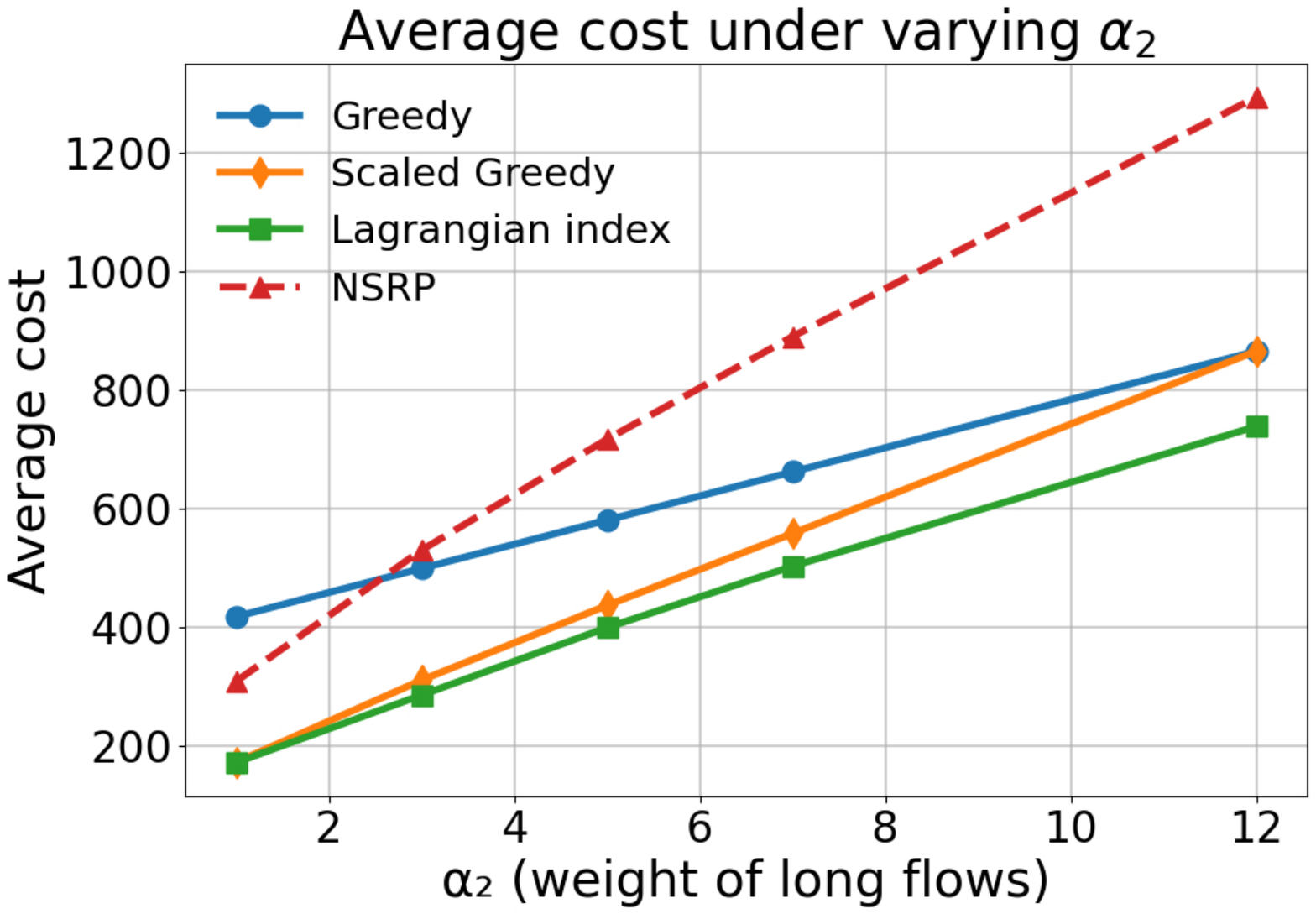}
        \caption{$p=0.7$}
        \label{fig:aoi_alpha_different_p07}
    \end{subfigure}
    \hfill
    \begin{subfigure}[t]{0.48\linewidth}
        \centering
        \includegraphics[width=\linewidth]{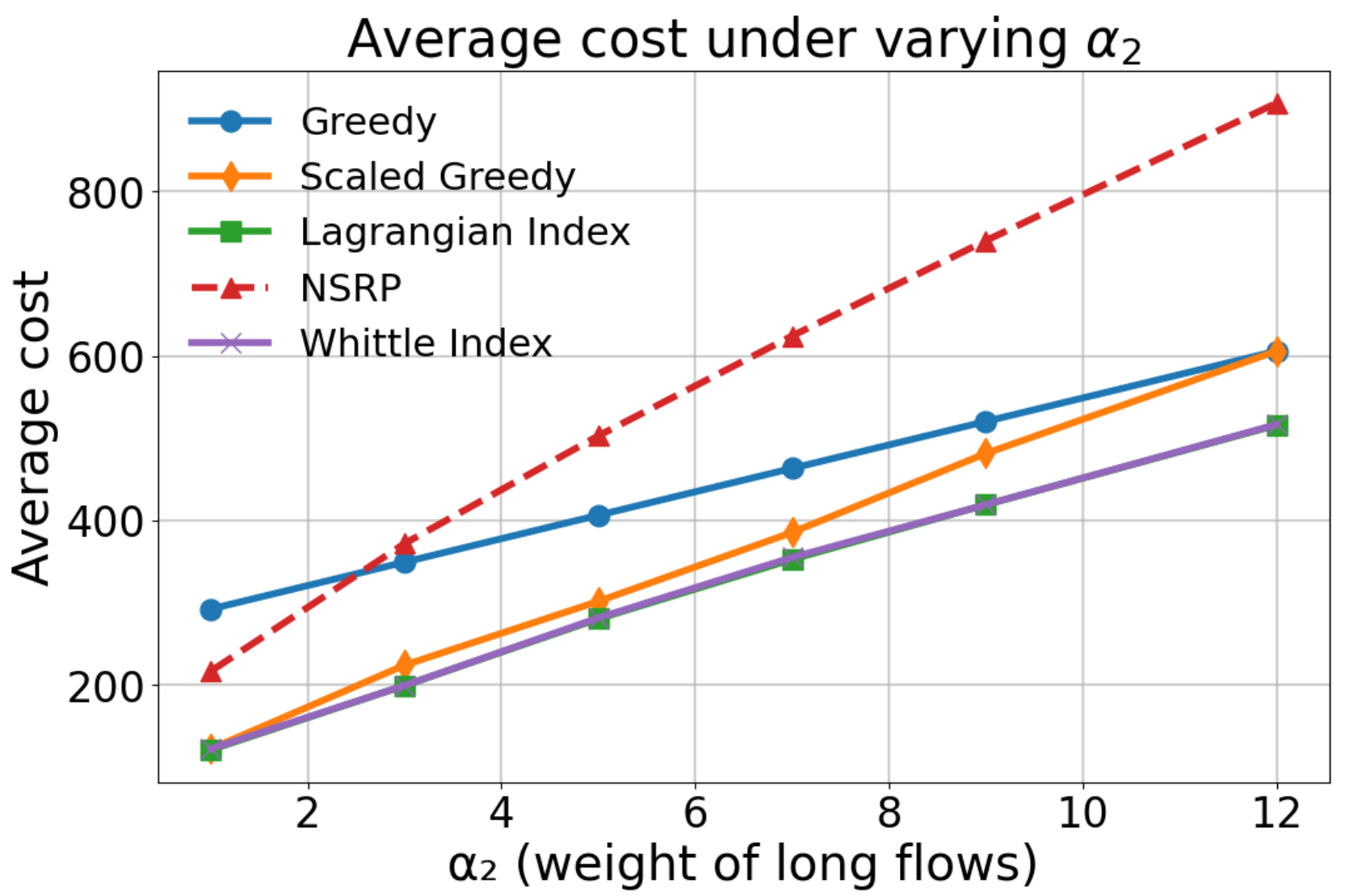}
        \caption{$p=1$}
        \label{fig:aoi_alpha_different_p1}
    \end{subfigure}
    \caption{
    Average weighted AoI versus $\alpha_2$. We consider $N=10$ flows, with five class--1 sources having $(L_1,\alpha_1)=(2,12)$ and five class--2 sources having $(L_2,\alpha_2)=(15,\alpha_2)$, where $\alpha_2 \in \{1,3,5,7,9\}$. The left panel corresponds to $p=0.7$, and the right panel to $p=1$. The Whittle index and Lagrange index policies achieve nearly identical performance.
    }
\label{fig:aoi_alpha_different_comparison}
\end{figure}

\section{Conclusion}
\label{Sec:conclusion}

In this paper, we developed a Lagrange index heuristic to minimize the weighted average AoI. The problem was formulated as a RMAB with SMDP dynamics. We proposed an heuristic based on the Lagrange index policy. We showed that the proposed policy outperforms both the NSRP in \cite{zhao} and a scaled greedy policy, demonstrating that structural analysis translates into tangible performance gains. For the weighted AoI model, we established key structural properties, most notably threshold behavior and leveraged them to derive efficient algorithms to compute the Lagrange indices. Beyond the specific AoI setting, the proposed framework extends naturally to a broader class of RMAB problems with SMDP dynamics, thereby expanding the scope of index-based control beyond standard discrete-time formulations. Future work will consider systems with switching and alternative update-generation models, building on the structural foundations developed here.
\bibliography{Ref}
\bibliographystyle{ieeetr}
\appendix
\subsection{Proof of Lemma~\ref{lem:h_monotone_full}}
\label{app:bias_monotone}

\begin{proof}

\textbf{Step 1: Convergence of RVI.}

Consider the single–flow SMDP under a fixed multiplier $\lambda$.
Let $\pi^{(i)}$ denote the stationary policy that schedules flow $i$
at every decision stage.

Under this policy, the AoI process $\{v(k)\}$ evolves as follows:
\[
v \to 
\begin{cases}
v+1, & \text{with probability } (1-p),\\
l \ge L_i, & \text{with probability } p_l^{(i)}.
\end{cases}
\]
The state space is $\{L_i, L_i+1, \dots\}$.

Since from any state $v \ge L_i$ there is positive probability
of transitioning to any state $l \ge L_i$ with $p_l^{(i)} > 0$,
the induced Markov chain is irreducible on its state space.

We now establish positive recurrence using a Foster–Lyapunov drift argument.
Consider the Lyapunov function $V(v) = v$.
The conditional drift satisfies
\begin{align*}
\mathbb{E}[v(k+1) - v(k) &\mid v(k)=v]
\\&= (1-p)(1)
+ \sum_{l\ge L_i} p_l^{(i)} (l - v) \\
&= (1-p) + L_i - (1-p)
- v p.
\end{align*}

Since $\sum_{l\ge L_i} p_l^{(i)} = p$
and $\sum_{l\ge L_i} p_l^{(i)} l = L_i - (1-p)$,
the drift simplifies to
\[
\mathbb{E}[v(k+1) - v(k) \mid v(k)=v]
= L_i - p v.
\]

For sufficiently large $v$, the drift is strictly negative.
Therefore, the Markov chain satisfies the Foster–Lyapunov condition
and is positive recurrent. Hence the induced Markov Chain under the policy $\pi^{(i)}$ is unichain. 
Since the one–stage cost is nonnegative in $v$ \cite{sennott1998stochastic}, 
the average cost under $\pi$ is finite.
Under this unichain condition and boundedness of the relative value
differences, standard SMDP results imply that Relative Value Iteration (RVI)
converges (up to an additive constant) to a solution $h_i$ of
\eqref{eq:smdp_single_flow_bellman_Q} \cite{sennott1998stochastic}.


\medskip
\textbf{Step 2: Monotonicity of the iterates.}

We prove by induction that
\[
v_1 \le v_2
\;\Longrightarrow\;
h^{(k)}(v_1) \le h^{(k)}(v_2),
\qquad \forall k.
\]

\emph{Base case:}
$h^{(0)} \equiv 0$ is nondecreasing.

\medskip
\emph{Induction step:}
Assume $h^{(k)}$ is nondecreasing.
Let $v_1 \le v_2$.

The stage cost $g_i(v,\mathbf a)$ is nondecreasing in $v$.
We verify monotonicity separately for the two types of actions.

\medskip
\noindent
\emph{Case 1: $j=i$.}
\[
Q_{i,i}^{h^{(k)}}(v)
=
g_i(v,\mathbf e_i)
+
(\lambda-\theta)L_i
+
\sum_{l\ge L_i} p_l^{(i)}\, h^{(k)}(l).
\]
The last two terms are independent of $v$.
Since $g_i(v,\mathbf e_i)$ is nondecreasing in $v$,
it follows that
\[
v_1 \le v_2
\;\Longrightarrow\;
Q_{i,i}^{h^{(k)}}(v_1)
\le
Q_{i,i}^{h^{(k)}}(v_2).
\]

\medskip
\noindent
\emph{Case 2: $j\neq i$.}
\[
Q_{i,j}^{h^{(k)}}(v)
=
g_i(v,\mathbf e_j)
-
\theta L_j
+
\sum_{l\ge L_j} p_l^{(j)}\, h^{(k)}(v+l).
\]
The term $-\theta L_j$ is constant in $v$.
Since $g_i(v,\mathbf e_j)$ is nondecreasing in $v$
and $v\mapsto v+l$ is increasing,
the induction hypothesis implies
\[
v_1 \le v_2
\;\Longrightarrow\;
h^{(k)}(v_1+l)
\le
h^{(k)}(v_2+l),
\qquad \forall l.
\]
Taking expectation preserves order, hence
\[
Q_{i,j}^{h^{(k)}}(v_1)
\le
Q_{i,j}^{h^{(k)}}(v_2).
\]

Since $v\mapsto v'$ is increasing and $h^{(k)}$ is nondecreasing,
\[
v_1 \le v_2
\;\Longrightarrow\;
Q_{i,j}^{h^{(k)}}(v_1)
\le
Q_{i,j}^{h^{(k)}}(v_2),
\qquad \forall j.
\]

Taking minima preserves order, hence
\[
(\mathcal T h^{(k)})(v_1)
\le
(\mathcal T h^{(k)})(v_2).
\]
So,
\[
h^{(k+1)}(v_1)
\le
h^{(k+1)}(v_2).
\]

Thus $h^{(k+1)}$ is nondecreasing.

\medskip
By induction, all iterates are nondecreasing.
Since RVI converges pointwise (up to a constant),
and monotonicity is preserved under limits,
the limiting bias function $h_i$ is nondecreasing.
\end{proof}

\subsection{Proof of Lemma~\ref{lem:dominant_competitor}}
\label{app:dominant_competitor}
\begin{proof}

We first consider the corresponding discounted-cost problem with discount
factor $\beta\in(0,1)$. Let $V_i^\beta(v)$ denote the discounted value
function. The Bellman equation can be written as
\[
V_i^\beta(v)
=
(1-p)\beta V_i^\beta(v+1)
+
\min\Big\{
Q_{i,i}^\beta(v),\;
\min_{j\neq i} Q_{i,j}^\beta(v)
\Big\},
\]
where
\[
Q_{i,j}^\beta(v)
=
\begin{cases}
g_i(v,\mathbf e_i)
+
\lambda L_i
+
\beta\displaystyle\sum_{l\ge L_i} p_l^{(i)} V_i^\beta(l),
& j=i, \\[0.8em]
g_i(v,\mathbf e_j)
+
\beta\displaystyle\sum_{l\ge L_j} p_l^{(j)} V_i^\beta(v+l),
& j\neq i.
\end{cases}
\]

For competing actions $j\neq i$, and let $m = \arg\min_{j\neq i} L_j$, the immediate cost term satisfies
\[
g_i(v,\mathbf e_j)
=
\alpha_i \left( v L_j + w(L_j)\right),
\]
where $w(L_j)$ is increasing in $L_j$. Since, $L_m \le L_j$, then
\[
g_i(v,\mathbf e_m)
\leq    
g_i(v,\mathbf e_j),
\]
and
\[
\sum_{l\ge L_j} p_l^{(j)} V_i^\beta(v+l).
\]

Consider the ratio of the corresponding stage-duration probability mass functions:
\[
\frac{p_l^{(j)}}{p_l^{(m)}}
=
\frac{
\binom{l-2}{L_j-2}
}{
\binom{l-2}{L_m-2}
}
p^{L_j-L_m}(1-p)^{-(L_j-L_m)},
\qquad l \ge L_j.
\]
Using the factorial representation of the binomial coefficients,
\[
\frac{
\binom{l-2}{L_j-2}
}{
\binom{l-2}{L_m-2}
}
=
\frac{(L_m-2)!}{(L_j-2)!}
\prod_{k=0}^{L_j-L_m-1}(l-L_m-k),
\]
which is strictly increasing in $l$ over the support $l \ge L_j$.
Hence, $\{S^{(j)}\}$ is increasing in likelihood ratio order. It follows that
\[
S^{(m)} \le_{\mathrm{lr}} S^{(j)}
\quad \Rightarrow \quad
S^{(m)} \le_{\mathrm{st}} S^{(j)},
\]
see, \cite{shaked2007stochastic}.
The monotonicity of the discounted value function follows by the same order-preserving dynamic programming argument used for the average-cost Bellman operator. The only modification is the multiplicative factor $\beta\in(0,1)$ applied to continuation values, which preserves inequalities. Hence the Bellman operator remains monotone, and the discounted value function is nondecreasing in $v$. Since $V_i^\beta(\cdot)$ is nondecreasing, stochastic ordering preserves the
expectation, yielding
\[
\mathbb{E}\!\left[V_i^\beta(v+S^{(m)})\right]
\le
\mathbb{E}\!\left[V_i^\beta(v+S^{(j)})\right],
\]
which establishes the desired inequality.

Hence, for every $\beta\in(0,1)$,
\[
Q_{i,m}^\beta(v_i)
\le
Q_{i,j}^\beta(v_i),
\qquad
\forall v_i,\;\forall j\neq i,
\]
where $m\in\arg\min_{j\neq i} L_j$.

Thus, in the discounted problem, whenever flow $i$ is not scheduled, it is optimal to select the competing flow with the smallest expected stage duration.
Finally, by Blackwell optimality (see \cite{Puterman_book}), if a stationary policy is optimal for all discount factors sufficiently close to $1$, then it is also optimal for the average-cost problem. Since the dominance of action $m$ holds uniformly in $v$ and for all $\beta$, the same action remains optimal in the average-cost SMDP. This completes the proof.
\end{proof}

\subsection{Proof of Theorem~\ref{thm:threshold_structure}}
\label{app:threshold_formula_theta}

\begin{proof}

Fix $\lambda$ and flow index $i$. Let $m \neq i$ denote any competing action.
We use a one–step look ahead argument and the optimality principle \cite{Bertsekas}.

\medskip
\noindent
\textbf{Step 1: Cost of scheduling flow $i$.}

Suppose that at state $v$ we schedule flow $i$, and thereafter continue
scheduling flow $i$ at every decision stage. The corresponding
cost-to-go satisfies
\begin{align*}
C_i(v)
&=
(\alpha_i v + \lambda - \theta_{\lambda})L_i
+ \alpha_i w(L_i)  \\
&\quad
+ (1-p)h_i(v+1)
+ \sum_{l \ge L_i} p_l^{(i)} h_i(l).
\end{align*}

Scheduling of only flow $i$ admits the affine solution
\[
h_i(x)
=
\frac{\alpha_i L_i}{p}\,x
-
\frac{\theta_{\lambda} L_i}{p},
\qquad x \ge T_i(\lambda),
\]
up to an additive constant.

\medskip
\noindent
\textbf{Step 2: One–step look ahead to flow $m$.}

Now consider a deviating policy that schedules flow $m$ at state $v$
for exactly one decision epoch and then switches permanently to
scheduling flow $i$.

The corresponding cost-to-go is
\begin{align*}
C_m(v)
&=
(\alpha_i v - \theta_{\lambda})L_m
+ \alpha_i w(L_m)  \\
&\quad
+ (1-p)h_i(v+1)
+ \sum_{l \ge L_m} p_l^{(m)} h_i(v+l).
\end{align*}

\medskip
\noindent
\textbf{Step 3: Comparison.}

Scheduling flow $i$ is optimal at state $v$ if
\[
C_i(v) - C_m(v) \le 0.
\]

The common continuation term $(1-p)h_i(v+1)$ cancels.
Substituting the affine form of $h_i(\cdot)$ and simplifying yields
a linear inequality in $v$ of the form
\[
v
\ge
\frac{\theta_{\lambda}}{\alpha_i}
-
\frac{L_m-1}{2p}
-
\frac{L_i}{p}.
\]

Define
\[
T_{i}(\lambda)
=
\left\lceil
\frac{\theta_{\lambda}}{\alpha_i}
-
\frac{L_m-1}{2p}
-
\frac{L_i}{p}
\right\rceil.
\]

Then for all $v \ge T_{i}(\lambda)$,
\[
C_i(v) \le C_m(v).
\]

\medskip
\noindent
\textbf{Step 4: Invariance and optimality.}

For $v \ge T_{i}(\lambda)$, if flow $m$ is scheduled,
the next state equals $v+1$ with probability $(1-p)$
or $v+l$ for some $l \ge L_m$.
In all cases, the next state is at least $v$,
and hence remains in the region $[T_{i,m}(\lambda),\infty)$.

Thus, once $v \ge T_{i,m}(\lambda)$, any one–step lookahead
to action $m$ cannot improve the cost,
and the process remains in the region where scheduling $i$
is better.

By the one–step deviation principle for average cost\cite{Bertsekas},
no profitable deviation exists for $v \ge T_{i}(\lambda)$.
we conclude that for all $v \ge T_i(\lambda)$,
scheduling flow $i$ is optimal,
whereas for $v < T_i(\lambda)$ scheduling flow $m \neq i$ is optimal.

Hence the optimal policy is of threshold type.

\end{proof}

\subsection{Proof of Lemma~\ref{lem:monotonicity_in_lambda}}
\label{app:monotonicity_in_lambda}
\begin{proof}
For any fixed threshold policy $T$, let $\theta_T(\lambda)$ denote the long-run average cost obtained when threshold $T$ is used under multiplier $\lambda$.

Under a fixed threshold policy, the multiplier $\lambda$ enters the average cost only through the penalty term associated with transmissions. Since the average transmission rate under any threshold policy is nonnegative, it follows that $\theta_T(\lambda)$ is nondecreasing in $\lambda$. Hence, for any $\lambda_1 \leq \lambda_2$,
\[
\theta_T(\lambda_1)\leq \theta_T(\lambda_2), \qquad \forall T.
\]

Now, the optimal average cost is obtained by minimizing over all feasible thresholds:
\[
\theta(\lambda)=\min_T \theta_T(\lambda).
\]
Therefore,
\[
\theta(\lambda_1)
=
\min_T \theta_T(\lambda_1)
\leq
\min_T \theta_T(\lambda_2)
=
\theta(\lambda_2),
\]
which proves that $\theta(\lambda)$ is nondecreasing in $\lambda$.

Next, from the threshold characterization,
\begin{equation}
T_i(\lambda)
=
\left\lceil
\frac{\theta(\lambda)}{\alpha_i}
-
\frac{L_m-1}{2p}
-
\frac{L_i}{p}
\right\rceil.
\label{eq:threshold_monotone_appendix}
\end{equation}
Since $\theta(\lambda)$ is nondecreasing in $\lambda$, the quantity inside the ceiling operator in \eqref{eq:threshold_monotone_appendix} is also nondecreasing in $\lambda$. The ceiling function preserves monotonicity, and therefore
\[
T_i(\lambda_1)\leq T_i(\lambda_2), \qquad \lambda_1\leq \lambda_2.
\]

Hence, both $\theta(\lambda)$ and $T_i(\lambda)$ are nondecreasing in $\lambda$.
\end{proof}

\end{document}